%% file: ms.tex
\documentclass[journal,comsoc]{IEEEtran}

\usepackage{amssymb, amsmath, epsfig, booktabs, amsthm, epstopdf, color, cite, bm}
\usepackage{graphicx}
\usepackage{subfig,float,stfloats}
\usepackage{psfrag}
\usepackage{paralist}
\usepackage{color}
\usepackage{mathrsfs}
\usepackage{xcolor}
\usepackage[normalem]{ulem} 
\usepackage[cmintegrals]{newtxmath}

\allowdisplaybreaks[4]

\newtheorem{theorem}{Theorem}

\newtheorem{lemma}{Lemma}

\newtheorem{remark}{Remark}

\newtheorem{corollary}{\bf Corollary}

\newtheorem{algorithm}{Algorithm}

\newcommand{\ied}{{i.e.}, }
\newcommand{\eqdef}{ := }
\renewcommand{\footnoterule}{%
  \kern -3pt
  \hrule width 150pt height .2pt
  \kern 2pt
}

\begin{document}
%
\title{\huge Performance of Multi-Cell Massive MIMO Systems With \\ Interference Decoding}
%
%
%

\author{Meysam~Shahrbaf~Motlagh,\IEEEmembership{}
        Subhajit~Majhi,~\IEEEmembership{}Patrick~Mitran~\IEEEmembership{}
\thanks{The authors are with the Department of Electrical and Computer Engineering, University of Waterloo, Waterloo, ON N2L 3G1, Canada.}}

%
%

{}
%



\maketitle

\begin{abstract}
We consider a multi-cell massive MIMO system where a time-division duplex protocol is used to estimate the channel state information via uplink pilots. When maximum ratio combining (MRC) is used at the BSs, the re-use of pilots across cells causes the pilot contamination effect which yields interference components that do not vanish as the number of base-station (BS) antennas $M \rightarrow \infty$. When treating interference as noise (TIN), this phenomenon limits the performance of multi-cell massive MIMO systems. In this paper, we analyze more advanced schemes based on simultaneous unique decoding (SD) as well as simultaneous non-unique decoding (SND) of the interference that can provide unbounded rate as $M \rightarrow \infty$. We also establish a worst-case uncorrelated noise technique for multiple-access channels to derive achievable rate expressions for finite $M$. Furthermore, we study a much simpler subset of SND (called S-SND) which provides a lower bound to SND and achieves unbounded rate as $M \rightarrow \infty$, and also outperforms SD for finite $M$. For the special cases of two-cell and three-cell systems, using a maximum symmetric rate allocation policy we compare the performance of different interference decoding schemes with that of TIN. Finally, we numerically illustrate the improved performance of the proposed schemes.
\end{abstract}


%
\IEEEpeerreviewmaketitle

\section{Introduction}
%
%
%
%
\IEEEPARstart{W}{ireless} communication standards are rapidly evolving to deal with challenges such as the ever increasing number of users as well as the demand for higher data rates and energy efficiency. These challenges give rise to the need to incorporate new protocols and techniques in 5G cellular networks. For instance, the METIS 5G project has as overall technical goals to increase the typical user data rate in a mobile network by 10x to 100x, handle 1000x more mobile data traffic per unit area, and support 10x to 100x more connected devices, all by 2020 \cite{GreenTouch}. Also in light of the demand for increased energy efficiency, the GreenTouch initiative has, for example, aimed to reduce the net energy consumption in end-to-end communication networks by up to $98 \%$ by 2020, compared to 2010 \cite{METIS}. In order to meet these goals, several promising solutions have been proposed for 5G, including massive multi-input multi-output antenna systems \cite{larsson2014massive}, cloud-RAN \cite{peng2015system}, ultra-densification \cite{bhushan2014network} and millimeter wave communications \cite{niu2015survey, majhi_IC_journal}.

In a massive MIMO communication system, each BS utilizes a very large number of antennas, which allows for the simultaneous serving of several (single or multi-antenna) users, where the number of BS antennas is normally assumed to be significantly larger than the number of users. The introduction of massive MIMO technology dates back to the seminal work of Marzetta in \cite{marzetta2010noncooperative}, where it was shown that, under the assumption of independent Rayleigh fading channels and when the number of BS antennas grows to infinity, the effects of small-scale fading, intra-cell interference and additive noise all vanish due to channel hardening effects. In theory, massive MIMO has several advantages including high energy efficiency \cite{ngo2013energy}, high spectral efficiency and increased capacity through the aggressive spatial multiplexing of many users \cite{bjornson2014optimizing, larsson2014massive}, all enabled by the use of simple linear precoding/decoding techniques \cite{hoydis2013massive}. With regards to energy efficiency in particular, the work of \cite{ngo2013energy} has analytically shown that when perfect CSI is available at the BSs, the uplink transmit power of each user can be scaled inversely proportionally with the number of BS antennas, without any performance loss.

Successful implementation of massive MIMO in practice, however, depends heavily on the availability of knowledge of channel state information (CSI) at the BSs. Thanks to the use of time-division duplex (TDD) protocols as suggested in \cite{marzetta2010noncooperative} as well as uplink and downlink channel reciprocity, the BSs are able to estimate downlink channels using uplink pilots and later employ these estimates for both precoding and decoding purposes. In reality, the length of the channel coherence time is finite, and therefore the number of available orthonormal pilot sequences is limited. Consequently, in a multi-cell system, the same set of orthonormal pilots is shared across multiple cells. This, in turn, degrades the channel estimation performance for a user in one cell in that it will be contaminated by the channels of users in other cells whose pilots are not orthonormal to the first user. More specifically, the pilot contamination (PC) effect becomes a source of interference that precludes the logarithmic growth of users' achievable rate with the number of BS antennas \cite{marzetta2010noncooperative}. Hence, it is believed that PC constitutes a fundamental bottleneck in multi-cell massive MIMO systems \cite{elijah2016comprehensive}. Some notable exceptions to this belief are the works of \cite{yin2016robust, adhikary2017uplink, bjornson2018massive} that will be discussed in the sequel.
\subsection{Contributions}
In this paper, we assume the same orthonormal pilots are used across multiple cells, and take a different view of the inter-cell interference caused by PC. More specifically, we show that when $M \rightarrow \infty$ the use of more sophisticated schemes such as decoding the PC interference, rather than simply treating it as noise, allows one to attain unbounded rates, even in the presence of the PC effect.

We summarize the major contributions of this paper as follows:

\begin{itemize}
\item Using the capacity region obtained by simultaneous unique decoding of the desired signal and PC interference (i.e., employing SD as opposed to TIN), it is shown that when maximum ratio combining (MRC) is employed at BS, the per-user rate tends to infinity as $M \rightarrow \infty$.
\item It is shown that when decoding interference due to the PC, reusing the same pilots across cells (as opposed to using different pilots) is preferable as it requires decoding significantly fewer interference terms.
\item The benefits of using simultaneous non-unique decoding (SND) is investigated, which strictly contains regions SD and TIN. Moreover, a simplified subset of SND (S-SND) is studied, which is shown to be strictly larger than SD and also provides a lower bound to SND.
\item A worst-case uncorrelated noise technique for multiple access channels (MAC) is established that for finite $M$ yields achievable rate expressions over regions SD/SND/S-SND.
\item  The problem of maximum symmetric rate allocation (i.e., maximizing the minimum achievable rate) for TIN/SD/SND/S-SND is investigated. Some structural results are also presented for the two extreme regimes of high and low SINR. In particular, it is found that when the number of BS antennas is truly large the interference decoding schemes SD/SND achieve the same performance and also strictly outperform TIN.
\item For the special case of a two-cell system and assuming a symmetric geometry, it is shown that for relatively small values of $M$, the PC interference is ``weak" in that SND and TIN achieve the same rate and both of these strictly outperform SD and S-SND. Hence, one may choose TIN which is simpler to implement. Nevertheless, for large values of $M$ (beyond a threshold), the PC interference becomes ``strong" so that the interference decoding schemes SD/SND provide the same performance and both of these strictly outperform TIN. Hence, one may choose SD which is simpler to implement. Analytical conditions in terms of mutual information expressions under which these results hold are also found.
\item For the special case of a three-cell system, it is numerically shown that the use of SND can provide a strictly better performance compared to \textit{all the other schemes}.
\end{itemize}
One should note that the theoretical contributions in the first five items listed above as well as the analytical conditions for the two-cell case are valid regardless of the numerical results presented in Section V. It is only in Section V that specific values for the system parameters are chosen to numerically compare the performance of different schemes and hence validate the analytical findings of Sections III and IV. 
\subsection{Related Work}\label{sec:1}
In order to tackle the PC problem, systematic solutions have been extensively studied in the literature; some attempt to alleviate the PC effects by reducing its impact on the system performance \cite{huh2012achieving, fernandes2013inter, appaiah2010pilot, jose2011pilot} whereas others, given that some assumptions and requirements are satisfied, suggest schemes that completely eliminate PC and provide unbounded achievable rates in the asymptotic regime \cite{yin2016robust, adhikary2017uplink, van2018large, bjornson2018massive}. 
Specifically, in the framework of PC mitigation schemes, the use of time-shifted pilots is proposed in \cite{fernandes2013inter, appaiah2010pilot}, where in order to make sure that non-orthonormal pilots do not overlap in time, the location of pilots in frames shift so that transmission in different cells is done at non-overlapping times. In addition to time-shifted pilots, power allocation algorithms are also proposed in \cite{fernandes2013inter} and have shown to provide significant gains. In \cite{jose2011pilot}, a multi-cell MMSE based precoding is investigated for downlink to minimize the sum of the mean square error of signals received at the users in the same cell and the mean square interference seen by users in other cells. Unfortunately, in all these techniques the PC effect is only partially suppressed and thus the achievable rates do not grow without bound as the number of BS antennas is increased. 

In the line of works that construct asymptotically noise and interference free systems with infinite capacity, \cite{yin2016robust} considers a semi-blind channel estimation technique to separate the subspace of the desired user channels from the subspace of interfering channels due to PC. However, in order for complete elimination of the PC effect, this method requires that the channel coherence time goes to infinity. Unfortunately, this assumption is not true in practice. Other interesting works include \cite{adhikary2017uplink, van2018large}, where a large-scale fading decoding (LSFD) technique is used that eliminates PC interference with the help of a network controller resulting in achievable rates that scale as $\mathcal{O} (\log M)$. Therein, BS cooperation is required for the exchange of large-scale fading coefficients between BSs and the network controller, which results in backhaul overhead. Different from \cite{adhikary2017uplink, van2018large}, the work of \cite{bjornson2018massive} uses a multi-cell MMSE precoding/combining technique and assumes that pilot-sharing users must have asymptotically linearly independent covariance matrices. This assumption, however, may not always be true and also requires the knowledge of channel covariance matrices at the BSs.

In this paper, we do not view the interference caused by PC as a fundamental limitation in a multi-cell system due to treating it as noise. This is because TIN is known to be suboptimal in some scenarios \cite{bandemer2011interference, baccelli2011interference}, and hence the present work proposes more sophisticated schemes based on interference decoding. 

The rest of this paper is organized as follows. In Section II, we present the system model and describe the PC problem. In Section III, to combat the PC problem we propose two interference decoding schemes, i.e., SD and SND, as well as an achievable subset of SND (S-SND). Achievable rate expressions based on the worst-case uncorrelated noise technique are also derived in this section. The problem of maximum symmetric rate allocation is studied in section IV along with some structural results for the extreme regimes of high and low SINR as well as for the special cases of two-cell and three-cell systems. In Section V, we demonstrate numerical results, and finally, section VI concludes this paper.

\textit{Notation}: We use boldface upper and lower case symbols to represent matrices and vectors, respectively. An $M \times M$ identity matrix and an all-zero vector are denoted by $\pmb{I}_M$ and $\pmb{0}$, respectively. The superscripts $(.)^{T}$, $(.)^{\dag}$, $(.)^{\ast}$, and $(.)^{-1}$ denote the transpose, Hermitian transpose, conjugate and inverse operations. The notation $\textrm{diag} (\pmb{v})$ represents a diagonal matrix with elements $v[1], v[2], ...$ of vector $\pmb{v}$ along its main diagonal. The expressions $\mathbb{E} \left[ . \right]$ and $\textrm{var} \left[ \right]$ are used to denote mean and variance of a random variable, respectively, and $\mathcal{CN} (\pmb{m}, \; \pmb{R})$ denotes the circular symmetric complex Gaussian distribution with mean vector $\pmb{m}$ and covariance matrix $\pmb{R}$.
\section{Preliminaries}\label{sec:2}
\subsection{System Model}
We consider a multi-cell communication system with $L$ cells, where each cell has a BS equipped with $M$ antennas serving $K$ ($M \gg K$) single antenna users. Assuming a flat-fading model, the channel between the $M$ antennas of the BS in cell $j$ and the users in cell $l$ is described by
\begin{equation}\label{eq:1}
\pmb{G}_{jl}=\pmb{H}_{jl}\pmb{D}_{jl}^{1/2},
\end{equation}   
where $\pmb{H}_{jl}=[\pmb{h}_{j1l},\pmb{h}_{j2l}, ..., \pmb{h}_{jKl}] \in \mathbb{C}^{M\times K}$ is the channel matrix associated with the channel vectors $\pmb{h}_{jkl} \in \mathbb{C}^{M \times 1}$ of small-scale fading coefficients between antennas of the BS in cell $j$ and the $k^{th}$ user in cell $l$, and $\pmb{D}_{jl} = \textrm{diag}(\beta_{j1l}, \beta_{j2l}, ..., \beta_{jKl} )$ is the matrix of large-scale fading coefficients. One may rewrite \eqref{eq:1} as 
\begin{equation}
\pmb{g}_{jkl}= \sqrt{\beta_{jkl}} \pmb{h}_{jkl},
\end{equation}
which explicitly shows that the large-scale fading coefficient $\beta_{jkl}$, which models shadowing and path loss and is assumed to be known at the BS, is constant with respect to the index $m$ of the BS antenna (see Fig. \ref{fig:1}). The latter follows since the distance between the BS and a user is much larger than the spacing between the antennas of the BS. We also assume a block fading model where the large-scale fading coefficients $\beta_{jkl}$ are constant over many coherence time intervals $T$ and known a $priori$, whereas small-scale fading coefficients $h_{jkl}[m], m=1, ..., M$ are constant over one coherence interval, and drawn independently in each coherence interval with $\pmb{h}_{jkl} \sim \mathcal{CN} (\pmb{0}, \, \pmb{I}_M)$ (i.e., flat-fading model). 

Furthermore, we consider TDD operation such that reciprocity holds between uplink and downlink channels. We take the frequency re-use factor to be one, i.e., the whole frequency band is used in one cell, and re-used in all the adjacent cells. This assumption, in particular, entails a worst-case inter-cell interference.
\begin{figure}[t!]
\centering 
\footnotesize
\def\svgwidth{213pt} 
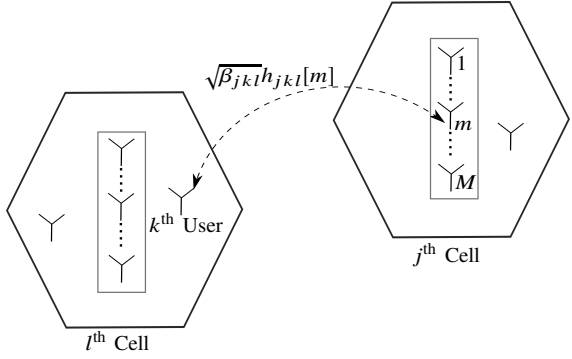 
\caption{\small System model showing the channel gain between the $m^{\textrm{th}}$ antenna of the BS in cell $j$ and the $k^{\textrm{th}}$ user in cell $l$.} 
\label{fig:1}
\end{figure}
\subsection{Uplink Data Transmission}
We point out that the model used for uplink data transmission in this paper is similar to that of \cite{hoydis2013massive} with slight change of notation. During the uplink data transmission phase, the BS in cell $j$ receives the baseband signal $\pmb{y}_{j} \in \mathbb{C}^{M \times 1}$ given by
\begin{equation}\label{eq:2}
\pmb{y}_{j} = \sum_{l=1}^{L}\nolimits \sum_{k=1}^{K}\nolimits \sqrt{\rho_{\textrm{u}}} \pmb{g}_{jkl} x_l [k] + \pmb{n}_j,
\end{equation}
where $\pmb{x}_l= \left[ x_l [1], x_l [2], ..., x_l [k] \right]^T$ is the vector of transmit signals of the users in cell $l$, $\rho_{\textrm{u}}$ is the average uplink transmit power of the users, and $\pmb{n}_j \sim \mathcal{CN} (\pmb{0}, \; \pmb{I}_M)$ is the additive Gaussian noise vector at the BS in cell $j$. Thus, $\rho_{\textrm{u}}$ can be interpreted as the uplink transmit SNR of the users. 
\subsection{CSI Estimation at BS}\label{sec:csi}
Similar to the approach of \cite{adhikary2017uplink} for CSI estimation, it is assumed that the same set of pilot sequences $\pmb{\psi}_1, \pmb{\psi}_2, ..., \pmb{\psi}_K \in \mathbb{C}^{\tau \times 1}$ of length $\tau$ (usually $\tau \geq K$, however without essential loss of generality we assume $\tau=K$) are used in all cells and thus the channel estimate will be corrupted by the PC from the adjacent cells. Defining the pilot matrix $\pmb{\Psi} = \left[ \pmb{\psi}_1, \pmb{\psi}_2, ..., \pmb{\psi}_K \right]^T \in \mathbb{C}^{K \times K}$, we assume orthonormal pilots $\pmb{\Psi} \pmb{\Psi}^{\dagger} = I_K$.

During the uplink training phase of the TDD protocol, user $k=1, 2, ..., K$  in each cell transmits the pilot sequence $\pmb{\psi}_k$ to its BSs. The BS in cell $j$ then finds the estimate $\pmb{\hat{G}}_{jj}$ of the local channels $\pmb{G}_{jj}$. More specifically, the BS in cell $j$ receives the matrix $\pmb{Y}_j^{\textrm{p}} \in \mathbb{C}^{M \times K}$, i.e.,
\begin{equation}\label{eq:4}
\pmb{Y}_j^{\textrm{p}} = \sum_{l=1}^{L} \nolimits \sqrt{\rho_{\textrm{p}}} \pmb{G}_{jl} \pmb{\Psi} + \pmb{Z}_j,
\end{equation}
where $\rho_{\textrm{p}}$ is the average pilot transmission power, and $\pmb{Z}_j$ is the AWGN at the BS with entries that are iid $\mathcal{CN} (0, 1)$ random variables. Similar to the uplink data transmission, $\rho_{\textrm{p}}$ can be interpreted as the pilot SNR. Generally, $\rho_\textrm{p}$ is a function of the average transmit power of users $\rho_{\textrm{u}}$ and the length of pilot sequences $\tau$. 

Multiplying $\pmb{Y}_j^{\textrm{p}}$ by $\pmb{\Psi}^{\dag}$, the $k^{\textrm{th}}$ column of the resulting matrix is
\begin{equation}\label{eq:5}
\pmb{r}_{jk} = \sum_{l=1}^L \nolimits \sqrt{\rho_{\textrm{p}}} \pmb{g}_{jkl} + \pmb{\tilde{z}}_{jk}, 
\end{equation} 
where $\pmb{\tilde{z}}_{jk} \sim \mathcal{CN} (\pmb{0}, \; \pmb{I}_M)$. The MMSE estimate $\pmb{\hat{g}}_{jkj}$ of $\pmb{g}_{jkj}$ based on the observation $\pmb{r}_{jk}$ is given by \cite{sengijpta1995fundamentals}
\begin{align}\label{eq:6}
\pmb{\hat{g}}_{jkj} &= \mathbb{E} \left[ \pmb{g}_{jkj} \pmb{r}_{jk}^{\dagger} \right] \mathbb{E} \left[ \pmb{r}_{jk} \pmb{r}_{jk}^{\dagger} \right]^{-1} \pmb{r}_{jk}  \\
&= \sqrt{\rho_{\textrm{p}}}\beta_{jkj} \left( 1 + \rho_{\textrm{p}} \sum_{l=1}^L \nolimits \beta_{jkl} \right) ^{-1} \pmb{r}_{jk}  \\
&= \alpha_{jkj} \left( \sum_{l=1}^L \nolimits \sqrt{\rho_{\textrm{p}}} \pmb{g}_{jkl} + \pmb{\tilde{z}}_{jk} \right), \label{eq:6:1}
\end{align}
where $\alpha_{jkj} := \frac{\sqrt{\rho_{\textrm{p}}} \beta_{jkj}}{1 + \rho_{\textrm{p}}\sum_{l=1}^L \beta_{jkl}}$. Due to the orthogonality property of MMSE estimation, one can decompose the channel $\pmb{g}_{jkj}$ as $\pmb{g}_{jkj} = \pmb{\hat{g}}_{jkj} + \pmb{\epsilon}_{jkj}$, where $\pmb{\epsilon}_{jkj}$ is the estimation error. It is also known that $\pmb{\epsilon}_{jkj}$ is uncorrelated (and thus independent due to the Gaussian assumption) from the estimate $\pmb{\hat{g}}_{jkj}$ \cite{adhikary2017uplink}. Therefore, it can be verified that $\pmb{\hat{g}}_{jkj} \sim \mathcal{CN} ( \pmb{0}, \; \sqrt{\rho_{\textrm{p}}} \beta_{jkj} \alpha_{jkj} \pmb{I}_M )$ and $\pmb{\epsilon}_{jkj} \sim \mathcal{CN} ( \pmb{0}, \; \beta_{jkj} ( 1 - \sqrt{\rho_{\textrm{p}}} \alpha_{jkj} ) \pmb{I}_M ) $. Using \eqref{eq:6}, one can see that the estimate $\pmb{\hat{g}}_{jkl}$ can be written in terms of $\pmb{\hat{g}}_{jkj}$ as follows
\begin{equation}\label{eq:7}
\pmb{\hat{g}}_{jkl} = \left( \frac{\beta_{jkl} }{ \beta_{jkj} } \right) \pmb{\hat{g}}_{jkj}.
\end{equation} 
\subsection{Treating Interference as Noise (TIN)}
Assuming maximum ratio combining (MRC), from \eqref{eq:2} the estimate of the $i^{\textrm{th}}$ user's signal in cell $j$ is
\begin{align}\label{eq:8}
\hat{y}_{ji} &= \pmb{\hat{g}}_{jij}^{\dag} \pmb{y}_{j}  \\
\nonumber &= \sum_{l=1}^{L} \nolimits \sqrt{\rho_{\textrm{u}}} \pmb{\hat{g}}_{jij}^{\dag} \pmb{g}_{jil} x_l [i]  \\ 
&\hspace{3mm}+ \sum_{l=1}^L  \nolimits \sum_{k=1,k \neq i}^K \nolimits \sqrt{\rho_{\textrm{u}}} \pmb{\hat{g}}_{jij}^{\dag} \pmb{g}_{jkl} x_l [k] + \pmb{\hat{g}}_{jij}^{\dag} \pmb{n}_j. \label{eq:8:1}
\end{align} 

Substituting  \eqref{eq:6:1} in  \eqref{eq:8:1}, and applying the strong law of large numbers, as $M \rightarrow \infty$, the following is obtained \begin{equation}\label{eq:9}
\frac{\hat{y}_{ji} }{ M } \stackrel{\rm a.s.}{\longrightarrow} \sqrt{\rho_{\textrm{u}}\rho_{\textrm{p}} } \alpha_{jij} \left( \beta_{jij} x_j [i] + \sum_{l=1, l \neq j}^L \nolimits \beta_{jil} x_l [i] \right),
\end{equation}
where $\stackrel{\rm a.s.}{\longrightarrow}$ represents the almost sure convergence. Note that a channel hardening effect is observed in \eqref{eq:9}.

Assuming TIN in uplink, the BS in cell $j$ only decodes the desired signal $x_{j} [i]$ and treats the remaining interfering signals $x_{l} [i], \; l \neq j$ as noise. Thus, defining $R_{ij}$ as the uplink rate of the $i^{\textrm{th}}$ user in cell $j$, any rate tuple $\left( R_{i1}, ..., R_{iL} \right)$ is achievable if it satisfies the following set of inequalities
\begin{equation}\label{eq:12:1}
R_{ij} \leq I \left( \hat{y}_{ji}  ; x_j [i] \Big\vert \; \pmb{\hat{g}}_{jij} \right), \quad \text{for} \quad  j=1, ..., L.
\end{equation}
Based on \eqref{eq:9}, it is known that \cite{marzetta2010noncooperative}
\begin{equation}\label{eq:12:2}
I \left( \hat{y}_{ji}  ; x_j [i] \Big\vert \; \pmb{\hat{g}}_{jij}  \right) \stackrel{\textrm{a.s.}}{\longrightarrow} C \left( \dfrac{\beta_{jij}^2}{\sum_{l=1,l \neq j}^L \beta_{jil}^2} \right), \; \textrm{as} \quad \; M \rightarrow \infty,
\end{equation}
where $C(x):= \log (1+x)$ is the Shannon rate function. 
\begin{remark}\label{remark:1}
\normalfont As the number of BS antennas $M$ in \eqref{eq:9} becomes large, except for the terms associated with the pilot-sharing users, i.e., the interference caused by PC, the effects of interference and noise vanish. 
\end{remark}
Throughout the paper, it is assumed that the noisy channel estimates $\pmb{\hat{g}}_{jij}$ are known locally at the BSs. Thus, from now on to simplify notation they will be omitted from the mutual information expressions.
\section{Decoding the PC Interference}
One can see that, as the number of BS antennas $M$ becomes large enough, the expression of \eqref{eq:12:2} converges to a constant independent of $M$ and thus the benefits of increasing $M$ saturate. In other words, treating interference due to PC as noise results in a fundamental limitation that constitutes a major bottleneck in overall performance of massive MIMO systems \cite{marzetta2010noncooperative}.

In this paper, as opposed to simply performing TIN, we consider more sophisticated schemes based on interference decoding. More specifically, we treat the PC interference terms as individual users (similar to a MAC) and thus try to decode them. As will be seen in the subsequent part, this change of perspective results in new achievable rate expressions that grow without bound as $M \rightarrow \infty$.
\subsection{Simultaneous unique Decoding (SD)}
Note that in the expression of the received signal after performing MRC in \eqref{eq:9}, the first term is the desired signal and the remaining non-vanishing terms are all inter-cell interference caused by users in other cells that are sharing the same pilot sequence, $\pmb{\psi}_i, i=1,\ldots,K$, as the $i^{\rm th}$ user of cell $j$. Now, consider \eqref{eq:8:1} which is the output of the $j^{\rm th}$ BS after performing MRC. If $\hat{y}_{ji}$, for $j=1, 2, ..., L, \; i=1, ..., K$, are considered together, then these represent the output of $K$ separate/non-interfering $L$-user interference channels (ICs), one such $L$-user IC for each pilot sequence $\pmb{\psi}_i, i=1,\ldots,K$: each $L$-user IC consists of $L$ transmitters, i.e., $i^{\rm th}$ user of all cells that are using the same pilot sequence $\pmb{\psi}_i$, and $L$ receivers, i.e., the BSs. One should also note that at each of the $L$ receivers of each IC, an asymptotically noise-free $L$-user MAC is observed (see \eqref{eq:9}). For instance, in the noise-free $L$-user MAC of \eqref{eq:9}, by unique joint decoding of the users  $\left[ x_1[i], x_2[i], ..., x_L [i] \right]^T$, unbounded rates are obtained as $M \rightarrow \infty$.
\begin{remark}
\normalfont Note that since large-scale fading coefficients from contaminating users are unknown at the BS, and also the effective channel gains in the MAC of \eqref{eq:9} are functions of these coefficicents, this MAC can be regarded as a compound MAC \cite{sabharwal2008compound}, where the channel gains from users to the receiver are unknown. It has been shown in \cite{sabharwal2008compound} that the achievable rates of a compound MAC (\ied a MAC with unknown channel gains) are the same as those of the standard MAC, where all channel gains are known. Therein, it has been shown that the lack of knowledge of channel gains at the receiver \textit{does not affect} the achievable rates, i.e., the users' signal can still be successfully decoded. 
\end{remark}

Due to finite coherence time of wireless channels, the number of available orthonormal pilot sequences is smaller than $KL$ for typical values of $K$ and $L$. Thus, one way to address this issue is to re-use the \textit{same} set of orthonormal pilots across all cells as described in Section \ref{sec:csi}. However, an alternative approach to that of Section \ref{sec:csi} is to use \textit{different} sets of orthonormal pilots in different cells. To illustrate this alternative, assume that a single set of orthonormal pilots is picked for one cell, and different rotated versions of this set are used in all other cells. In particular, user $k$ in cell $l$ transmits the pilot sequence $\pmb{\psi}_{kl}$ to its BS, where the entire pilot matrix used in cell $l$ is denoted by $\pmb{\Psi}_l$. As the sequences of other pilot matrices, $\pmb{\Psi}_j, \; j \neq l \in \lbrace 1, 2, ..., L \rbrace$, are rotated versions of sequences in $\pmb{\Psi}_l$, they have non-zero inner product.

After transmission of all pilot sequences, the BS in cell $j$ receives the matrix $\pmb{Y}_j^{\textrm{p}} \in \mathbb{C}^{M \times K}$, given by
\begin{equation}\label{eq:14:14}
\pmb{Y}_j^{\textrm{p}} = \sum_{l=1}^{L} \nolimits \sqrt{\rho_{\textrm{p}}} \pmb{G}_{jl} \pmb{\Psi}_l + \pmb{Z}_j.
\end{equation}
Multiplying $\pmb{Y}_j^{\textrm{p}}$ by $\pmb{\Psi}_j^{\dag}$, the $k^{\textrm{th}}$ column of the resulting matrix is
\begin{equation}\label{eq:14:14:1}
\pmb{r}_{jk} =  \sqrt{\rho_{\textrm{p}}} \pmb{g}_{jkj}  + \sum_{l=1, l \neq j}^L \nolimits \sqrt{\rho_{\textrm{p}}} \pmb{G}_{jl} \pmb{\Psi}_l \pmb{\psi}_{kj}^{\dagger} + \pmb{q}_{jk}, 
\end{equation} 
where $ \pmb{q}_{jk} \sim \mathcal{CN} (\pmb{0}, \; \pmb{I}_M)$. Therefore, the MMSE estimate $\pmb{\hat{g}}_{jkj}$ of $\pmb{g}_{jkj}$ based on the observation $\pmb{r}_{jk}$ is 
\begin{align}
\nonumber \pmb{\hat{g}}_{jkj} &= \mathbb{E} \left[ \pmb{g}_{jkj} \pmb{r}_{jk}^{\dagger} \right] \mathbb{E} \left[ \pmb{r}_{jk} \pmb{r}_{jk}^{\dagger} \right]^{-1} \\
&\hspace{3mm}\times \left(  \sqrt{\rho_{\textrm{p}}} \pmb{g}_{jkj}
+ \sum_{l=1, l \neq j}^L \nolimits \sqrt{\rho_{\textrm{p}}} \pmb{G}_{jl} \pmb{\Psi}_l \pmb{\psi}_{kj}^{\dagger} + \pmb{q}_{jk} \right). \label{eq:14:14:2}
\end{align}
One can readily see from \eqref{eq:14:14:2} that the channel estimate $\pmb{\hat{g}}_{jkj}$ is now contaminated by the channel of all the users in other cells. Thus, after performing MRC and letting $M \rightarrow \infty$, the non-vanishing terms in \eqref{eq:9} will include the signal of every user in every other cell. In turn, when employing interference decoding schemes, using the same set of pilots in different cells results in decoding $L$ users, whereas using different sets of pilots in different cells, as explained above, results in decoding $K (L-1) + 1$ users. As will be explained later in Remark~\ref{remark:different:pilots}, this alternative approach that requires decoding $K (L-1) + 1$ users (instead of $L$ users) degrades the performance of interference decoding schemes, as compared to the approach of Section \ref{sec:csi}. Moreover, the complexity of jointly decoding $K (L-1) + 1$ users is larger than that of decoding $L$ users.  Hence, when decoding the PC interference, using the same set of pilots in different cells (as opposed to different pilots) is preferable as it results in fewer interference terms to be decoded.

We now provide a detailed analysis of the achievable rates for finite values of $M$. Following the approach of \cite{adhikary2017uplink}, by adding and subtracting a term associated with the mean of the effective channel $\pmb{\hat{g}}_{jij}^{\dag} \pmb{g}_{jil}$ in \eqref{eq:8}, the following is obtained over one coherence interval
\begin{align}
\nonumber \hat{y}_{ji} &= \underbrace{ \sqrt{\rho_{\textrm{u}}} \sum_{l=1}^L \nolimits \mathbb{E} \left[ \pmb{\hat{g}}_{jij}^{\dag} \pmb{g}_{jil} \right] x_l [i]}_{\textrm{Desired signals}} \\
\nonumber &\hspace{4mm}+ \underbrace{\sum_{l=1}^{L} \nolimits \sqrt{\rho_{\textrm{u}}} \left(  \pmb{\hat{g}}_{jij}^{\dag} \pmb{g}_{jil} - \mathbb{E} \left[ \pmb{\hat{g}}_{jij}^{\dag} \pmb{g}_{jil} \right] \right) x_l [i]}_{\textrm{Interference due to channel estimation error}} \\
&\hspace{4mm}+ \underbrace{\sum_{l=1}^L \nolimits \sum_{k=1,k \neq i}^K \nolimits \sqrt{\rho_{\textrm{u}}} \pmb{\hat{g}}_{jij}^{\dag} \pmb{g}_{jkl} x_l [k]}_{\textrm{Interference caused by other users}} + \underbrace{\pmb{\hat{g}}_{jij}^{\dag} \pmb{n}_j }_{\textrm{Noise}} \label{eq:14} \\
&= \sum_{l=1}^L \nolimits \gamma_{il} x_l [i] + z_{jij}^{\prime}, \label{eq:15}
\end{align}
where $\gamma_{il} := \sqrt{\rho_{\textrm{u}}} \mathbb{E} [ \pmb{\hat{g}}_{jij}^{\dag} \pmb{g}_{jil} ]$ and $z_{jij}^{\prime}$ is the effective noise term.

The power of the desired signals in \eqref{eq:15} is proportional to $\left\vert \gamma_{il} \right\vert^2$ and is thus proportional to $M^2$. Moreover, the power of the effective noise term $z_{jij}^{\prime}$ is proportional to $M$. Therefore, by unique joint decoding of the users' signals $\left\lbrace x_l [i] \right\rbrace_{l=1}^L $ in \eqref{eq:15}, the achievable rates of the corresponding MAC grow unboundedly as $M \rightarrow \infty$.

Note that the effective noise $z_{jij}^{\prime}$ in \eqref{eq:15} contains the last three terms in \eqref{eq:14} including the inner product $\pmb{\hat{g}}_{jij}^{\dag} \pmb{n}_j$ of two Gaussian vectors, and hence is neither Gaussian nor independent of the users' signals. However, as will be shown in the sequel, it is uncorrelated from the users' signals. The following lemma lower bounds the mutual information terms defining the boundaries of the MAC region, using a Gaussian effective noise with the same power as that of $z_{jij}^{\prime}$ in \eqref{eq:15}.

\begin{lemma}\label{lemma:worst}
Consider the $L$-user MAC given by $y = \sum_{i=1}^L x_i^G + z$, where the users' signals $x_i^G, \; i=1, ..., L$ are independent with complex Gaussian distribution $x_i^G \sim \mathcal{CN} (0, \; P_i)$, and the additive noise $z$ is a complex random variable with mean zero and variance $\sigma_z^2$. If $z$ is uncorrelated from $x_i^G, \; i=1, ..., L$, then
\begin{equation}\label{eq:w}
 I \left( \pmb{x}_{\Omega}^G ; \; y^G \Big\vert \pmb{x}_{\Omega^c}^G \right) \leq I \left( \pmb{x}_{\Omega}^G ; \; y \Big\vert \pmb{x}_{\Omega^c}^G  \right),
\end{equation}  
where $\pmb{x}_{\Omega}^G$ is the vector with entries $ x_i^G, i \in \Omega \subseteq \lbrace 1, 2, ..., L \rbrace, \Omega \neq \emptyset$, $\Omega^c \eqdef \lbrace 1, 2, ..., L \rbrace \setminus \Omega$,   $y^G = \sum_{i=1}^L x_i^G + z^G$, and $z^G \sim \mathcal{CN} (0, \; \sigma_z^2)$. 
\end{lemma}
\begin{proof}
See Appendix A.
\end{proof}
Note that using Lemma \ref{lemma:worst} one can obtain an achievable lower bound on the capacity of a MAC with uncorrelated additive non-Gaussian noise by replacing the noise term with an independent zero mean Gaussian noise having the same variance. This is known as the worst-case uncorrelated noise result which has been previously established for the case of a point-to-point channel in the literature \cite{hassibi2003much}. When the additive noise is independent of the users' signals, Gaussian noise has been proven to be the worst-case noise for point-to-point, MAC, degraded broadcast and MIMO channels \cite{el2011network}. However, the proof provided in Lemma \ref{lemma:worst} only requires the additive noise to be uncorrelated of the users' signals.  

We now consider the MAC of \eqref{eq:15} at BS $j$. Using the usual definitions as in \cite{cover2012elements}, each message $m_l \in [ 1 : 2^{n R_{il}} ],\; l=1, ..., L $ (distributed uniformly) is encoded into the codeword $\pmb{x}_l^{n} [i] (m_l)$ of length $n$ which is generated iid $\mathcal{CN} (0, \; 1)$. Using SD and the standard random coding analysis as in \cite{cover2012elements}, it can be shown that decoding error probability tends to zero as $n \rightarrow \infty$, i.e., the rate tuple $\left( R_{i1}, ..., R_{iL} \right)$ is achievable, if
\begin{equation}\label{eq:16}
R_{\Omega} := \sum_{l \in \Omega} \nolimits R_{il} \leq I \left(  \hat{y}_{ji} ; \; \pmb{x}_{\Omega} \; \Big\vert \; \pmb{x}_{\Omega^c} \right),
\end{equation}
where $S= \lbrace 1, 2, ..., L\rbrace$, $\Omega \subseteq S$, and $\pmb{x}_{\Omega}$ is the vector with entries $ x_l [i], l \in \Omega$. Finally, to obtain the achievable region network-wide (at all BSs), one should take the intersection of achievable regions over all BSs. It should be pointed out that recent studies in the literature have proposed practical schemes using off-the-shelf LDPC codes that can achieve a performance very close to the theoretical SD \cite{zhu2017gaussian, sula2019compute}.

Note that $x_l [i]$ and $x_j [k]$ are independent for $(l, i) \neq (j, k)$, $\pmb{g}_{jil}$ and $\pmb{g}_{mkn}$ are independent for $(j,i,l) \neq (m, k, n)$, and also $\pmb{\hat{g}}_{jil}$ and $\pmb{\epsilon}_{jil}$ are uncorrelated. Therefore, for transmission over multiple coherence intervals all interference and noise terms in \eqref{eq:14} are uncorrelated from the desired signal components. Thus applying Lemma \ref{lemma:worst}, an achievable lower bound to $ I (  \hat{y}_{ji} ; \; \pmb{x}_{\Omega} \; \big\vert \; \pmb{x}_{\Omega^c} )$ in \eqref{eq:16} is obtained, similar to the ones established in \cite{adhikary2017uplink}. More specifically, replacing the effective noise $z_{jij}^{\prime}$ in \eqref{eq:15} by an independent Gaussian noise with a variance equal to the sum of the variances of the interference and noise terms in \eqref{eq:14}, provides a lower bound in \eqref{eq:16}. This is formally presented in the following theorem. 
\begin{theorem}\label{lemma:2}
Assuming $\pmb{x}_l= \left[ x_l [1], x_l [2], ..., x_l [K] \right]^T \sim \mathcal{CN} \left( \pmb{0}, \; \pmb{I}_K \right)$ for $l \in \lbrace 1, 2, ..., L \rbrace $, the following set of lower bounds can be achieved for the MAC given in \eqref{eq:15} at BS $j$
\begin{align}
I \left(  \hat{y}_{ji} ; \; \pmb{x}_{\Omega} \; \Big\vert \; \pmb{x}_{\Omega^c} \right) \geq C \left( \frac{ P_1 }{P_2 + P_3 + P_4} \right):= C_{\rm LB}(\Omega), \label{eq:16:1}
\end{align} 
which holds for $\forall \Omega \subseteq \lbrace 1, 2, ..., L \rbrace$, where
\begin{align}
&P_1 = M^2 \sum_{l \in \Omega} \nolimits \rho_{\rm p} \rho_{\rm u}   \beta_{jil}^2 \alpha_{jij}^2, \label{eq:17:1} \\
&P_2 = M  \sqrt{\rho_{\rm p}} \beta_{jij} \alpha_{jij} \sum_{l=1}^L \nolimits \rho_{\rm u} \beta_{jil}, \label{eq:17:2} \\
&P_3 = M  \sqrt{\rho_{\rm p}} \beta_{jij} \alpha_{jij} \sum_{l=1}^L \nolimits \sum_{k=1, k \neq i}^K \nolimits  \rho_{\rm u}  \beta_{jkl}, \label{eq:17:3} \\ 
&P_4 = M \sqrt{\rho_{\rm p}} \beta_{jij} \alpha_{jij}. \label{eq:17:4}
\end{align}
\end{theorem}
\begin{proof}
See Appendix B.
\end{proof}
\begin{remark}
\normalfont Note that $P_1$ is the power of the desired signal components associated with $\pmb{x}_{\Omega}$, whereas $P_2$ is the power of the interference due to the channel estimation error, $P_3$ is the power of the interference of other users, and $P_4$ is the power of the noise.  
\end{remark}
Using \eqref{eq:17:1}-\eqref{eq:17:4}, one can simplify the right hand side of \eqref{eq:16:1} as
\begin{align}
C_{\rm LB}(\Omega) = C \left( \dfrac{ M \sum_{l \in \Omega} \sqrt{\rho_{\textrm{p}}} \rho_{\textrm{u}}   \beta_{jil} \alpha_{jil} }{   \sum_{l=1}^L \sum_{k=1}^K \rho_{\textrm{u}} \beta_{jkl} +1 } \right), \quad \forall \Omega \subseteq \lbrace 1, 2, ..., L \rbrace \label{eq:17:6}
\end{align}
which follows from the fact that $ \beta_{jil}  \alpha_{jij} = \beta_{jij} \alpha_{jil}$. Therefore, it is clear that $I (  \hat{y}_{ji} ; \; \pmb{x}_{\Omega} \; \Big\vert \; \pmb{x}_{\Omega^c} ) \geq C (M \times \kappa)$, where $\kappa$ is a function of $\rho_{\textrm{u}}$, $ \rho_{\textrm{p}}$ and large-scale fading coefficients, and is also constant. Thus, the uplink achievable rates in \eqref{eq:16:1} grow as $\mathcal{O} (\log M)$. 

It is important to note that the BS $j$ is only interested in correct decoding of $x_{j}  [i]$ in uplink. Thus, incorrectly decoding $x_l [i], \; l \neq j$, should not penalize the rates achievable at BS $j$. Furthermore, the power of received signal for the users located in distant cells is very small, and thus trying to decode signals of such users can reduce achieved rates considerably. As later illustrated in the paper, there exist scenarios where system performance is constrained by these distant cells, which motivates the need for more advanced decoding schemes.

\subsection{Simultaneous Non-unique Decoding (SND)}\label{subsection:SND}
In this part, we investigate the benefit of using SND and further show that it enlarges the region obtained by SD for finite $M$. The optimality of this decoding scheme for interference networks with point-to-point codes and time-sharing has been shown in \cite{bandemer2015optimal}. Associated with the estimate of $x_l [i]$ in cell $l$, we consider an IC that consists of the $L$ senders, i.e., the $i^{\textrm{th}}$ user in each of cell $l=1, ..., L$, and the $L$ BS receivers. In particular, the BS $j \in S = \lbrace 1, 2, ..., L \rbrace$ simultaneously decodes the intended message $x_j [i]$ and the interference signals $x_l [i], \; l \neq j$, where incorrect decoding of the interference signals does not incur any penalty. More precisely, BS $j$ finds the unique message $\hat{m}_j$ such that $ ( \hat{\pmb{x}}_j^{n} [i] (\hat{m}_j)$, $\pmb{\hat{x}}_{S \setminus \lbrace j \rbrace}^{n} [i] (m_{S \setminus \lbrace j \rbrace})$, $\pmb{\hat{y}}_{ji}^{n} )$ is \emph{jointly typical} for some $m_{S \setminus \lbrace j \rbrace}$, where $\pmb{\hat{x}}_{S \setminus \lbrace j \rbrace}^{n} [i] (m_{S \setminus \lbrace j \rbrace})$ is the tuple of all codewords $\pmb{\hat{x}}_l^{n}[i](m_l)$ for $l \in S\setminus \lbrace j \rbrace $. For a comprehensive treatment of random code ensembles and joint typicality, we refer the reader to \cite[Chap.~3]{cover2012elements}.

It has been shown in \cite{bandemer2015optimal} that, assuming point-to-point random code ensembles, the capacity region of the IC associated with the $i^{\rm th}$ users across the $L$ cells can be described by
\begin{equation}\label{eq:23}
\mathcal{C}_i = \bigcap_{j \in S} \mathcal{R}_{ji},
\end{equation} 
where $\mathcal{R}_{ji}$ is the rate region achievable at BS $j$ given by
\begin{equation}\label{eq:24}
\mathcal{R}_{ji} = \bigcup_{\lbrace j \rbrace \subseteq \Omega \subseteq S} \mathcal{R}_{\textrm{MAC}(\Omega,j)}^i,
\end{equation}
and $\mathcal{R}_{\textrm{MAC}(\Omega,j)}^i$ represents the achievable rate region obtained from unique joint decoding of the signals $x_l [i], \; l \in \Omega$ at BS $j$. Note that $\Omega$ at BS $j$ must contain the index of the desired signal $x_j [i]$. 

The rate region $\mathcal{R}_{\textrm{MAC}(\Omega,j)}^i$ has the following properties:
\begin{description}
\item[{[P1]}] The region does not include the rates $R_{il}, \; l \in  \Omega^c$, and is thus unbounded in these variables.
\item[{[P2]}] The signals $x_l [i], \; l \in \Omega^c$, are treated as noise in the rate expressions defining the region.
\end{description}
One can readily see that $\mathcal{R}_{ji}$ strictly contains the MAC region at BS $j$. Therefore, the capacity region $\mathcal{C}_i$ in \eqref{eq:23} (obtained from SND) is strictly larger than the intersection of the MAC regions (obtained from SD) at BSs $l=1, ..., L$. Another important observation is that, due to [P2], $\mathcal{R}_{ji}$ also contains the TIN region (a similar observation was also made in \cite{bandemer2015optimal} and \cite{bandemer2011interference}) and thus parts of region $\mathcal{R}_{ji}$'s boundary remain constant when $M \rightarrow \infty$. It is also worth mentioning that a low complexity technique, called sliding-window coded modulation (SWCM) has been recently proposed in the literature that can achieve a performance close to that of the theoretical SND, while outperforming TIN in the \textit{strong} interference regime \cite{wang2014sliding, park2014interference, kim2015adaptive, kim2016interference}. We will see in the next section that depending on the number of BS antennas and geometry of the cells, the use of SND automatically specifies the optimal subset of signals that should be jointly decoded while the remaining signals will be treated as noise. 

Note that in SD, the decoder attempts to uniquely decode the message tuple of all users (i.e., the intended one as well as the interfering users), as in a MAC. While in SND the decoder attempts to decode only the intended message uniquely and the messages of interfering users non-uniquely. More specifically, for SND the decoder needs to perform jointly typical decoding of all possible message tuples that include the message of the intended user (i.e., the intended message only, all $2$-message tuples containing the intended message, ..., all $(n-1)$-message tuples containing the intended message and the only $n$-message tuple) as in \eqref{eq:24}. Hence, the SND decoder is more complex than that of SD.
\begin{remark}
\normalfont Note that there exists a \textit{complexity-performance trade-off} between the two interference decoding schemes SND and SD, and also between SD/SND and TIN. As explained above, the SD scheme attempts to decode all users, including those that have weak interference. Hence, while SD requires \textit{less complexity} than SND, as will be seen later in the paper, it achieves \textit{worse rates} than SND. In contrast, SND is able to adaptively determine whether a user should be decoded or treated as noise based on the strength of the interference. Hence, even though SND has \textit{more complexity} than SD, it achieves \textit{larger rates} than SD. A similar trade-off exists between SD/SND and TIN. Specifically, while the proposed schemes of SD/SND have \textit{more complexity} than TIN as they need to decode additional users, for sufficiently large number of antennas $M$, the rates achieved by TIN \textit{saturate} to a fixed value that does not increase with $M$. In contrast, the rates for SND/SD \textit{increase} as $\mathcal{O}(\log M)$, and hence as $M \rightarrow \infty$, \textit{unbounded rates} are obtained. 
\end{remark}
 \begin{remark}\label{remark:SIC}
\normalfont Note that the successive interfernce cancelation (SIC) technique used in \cite{van2016joint} is different from the SND/SD of this paper in the following manner: the work of \cite{van2016joint} considers a setting in downlink where each user is served by all BSs through the reception of $L$ independent data symbols from the $L$ BSs. In particular, each user applies SIC to sequentially decode the $L$ intended data symbols transmitted by the BSs, while treating all interfering signals, including pilot-sharing interfering signals, as noise, thus resulting in the rate saturation problem. This is in contrast to the approach proposed in this paper. As the BSs try to jointly decode (either uniquely or non-uniquely) the intended signal along with the signal coming from the pilot-sharing users, there is no rate saturation as $M$ increases.   
\end{remark}

\subsection{A Simplified Subset of SND (S-SND)}
We now consider a simplified achievable region which is a subset of SND and also described in \cite[Eq. (6.5)]{el2011network}. We refer to this region as S-SND, which is given by the following set of inequalities at BS $j$
\begin{equation}\label{eq:25}
\sum_{l \in \Omega} \nolimits R_{il} \leq I \left(  \hat{y}_{ji} ; \; \pmb{x}_{\Omega} \; \Big\vert \; \pmb{x}_{\Omega^c}\right),
\end{equation}
for all $\Omega$ such that $\lbrace j \rbrace \subseteq \Omega \subseteq \lbrace 1, 2, ..., L \rbrace$.

One can directly verify that region S-SND can be obtained from SD by removing all $2^{L-1}-1$ inequalities in \eqref{eq:16} that do not involve the rate $R_{ij}$. Hence, the region SD is strictly contained in S-SND. Furthermore, due to Theorem \ref{lemma:2} it can be verified that the boundaries of S-SND in \eqref{eq:25} grow as $\mathcal{O} (\log M)$. 

The motivation behind considering this region is as follows. It will be shown in the next section that, as opposed to SND, the S-SND region is in the form of a convex polytope which makes it tractable for computing the maximum symmetric rate allocation. Therefore, even though for large networks (e.g., more than 3 cells) the maximum symmetric rate of SND can not be computed in a computationally efficient way, it is feasible under S-SND. Furthermore, since S-SND is a subset of SND, it provides a lower bound to SND. As will be shown later, there are cases where S-SND strictly outperforms other schemes (e.g., it strictly outperforms SD in the low SINR regime). Thus, based on these findings we are able to draw conclusions regarding the performance of SND.
\section{Maximum Symmetric Rate Allocation}
Considering \eqref{eq:15}, it is evident that users with relatively small effective channel gains $\gamma_{il}, \; l \neq i,$ suffer from smaller rates compared to those users with stronger channels. Therefore, it is crucial to assure fairness among users when allocating resources in cellular networks. As such, we study the problem of maximum symmetric rate allocation policy (which is the same as maximizing the minimum achievable rate among all users) for various schemes. More specifically, we will compare the performance of all interference decoding schemes SD/SND/S-SND with that of TIN based on the maximum symmetric rate they can offer. In what follows, the analysis is shown only for the $i^{\rm th}$ ($i$ is arbitrary) users across multiple cells that are employing the same pilots, since the same results hold for other sets of pilot-sharing users.

A rate allocation is said to be symmetric when all users are assigned the same rate. Thus, the maximum symmetric rate associated with BS $j$ is obtained by $R_{\textrm{Sym}}^j = \max R$ such that the rate vector $\pmb{R} = \left[ R, R, ..., R \right]^T$ belongs to the achievable region at BS $j$. Therefore, the rate vector $[ R_{\textrm{Sym}}^j, R_{\textrm{Sym}}^j, ..., R_{\textrm{Sym}}^j ]^T$ must lie at the intersection of the diagonal $\left( R_{i1} = ... = R_{iL} \right)$ with the boundary of the achievable region at BS $j$.

One can verify that the SD region described in \eqref{eq:16} (achieved at BS $j$) can be represented as the intersection of a finite number of closed half-spaces and is also bounded. Hence, it is a convex polytope, denoted by $\mathscr{R}_j$, shown below
\begin{equation}\label{eq:26:1}
\mathscr{R}_j = \left\lbrace \left[ R_{i1}, ..., R_{iL} \right]^T \hspace{-1mm}:\hspace{-1mm}\sum_{l \in \Omega} \nolimits  R_{il} \leq g_{j}(\Omega), \forall \Omega \subseteq \left\lbrace 1, 2, ..., L \right\rbrace   \right\rbrace \hspace{-1mm},
\end{equation}
where the function $g_{j}(\Omega)$ is the r.h.s of the inequality in \eqref{eq:16}. Similarly, it can be verified that the region S-SND in \eqref{eq:25} is of the form \eqref{eq:26:1} except now $g_{j}(\Omega) = \infty$ if $ j \notin \Omega$, and is also a convex polytope.

The following lemma can be used to find the maximum possible value for the minimum entry of a vector $\pmb{R}$, where $\pmb{R} \in \mathscr{R}_j$. 
\begin{lemma}\label{lemma:fair} 
In the polytope $\mathscr{R}_j$, define
\begin{align}\label{min}
&\pi = \max \quad \min_{i \in S} \quad R_i \\
&\textup{subject to} \; \left[ R_1, ..., R_L \right]^T \in \mathscr{R}_j,
\end{align}
where $S = \lbrace 1, 2, ..., L \rbrace$. Then,
\begin{align} \label{min:1}
\pi =  \min_{\Omega \subseteq \left\lbrace 1, 2, ..., L \right\rbrace, \Omega \neq \emptyset} \frac { g_{j}(\Omega) }{   \vert \Omega \vert }.
\end{align}
\end{lemma}
\begin{proof}
Following the steps of \cite{fujishige2005submodular}, consider an arbitrary vector $\pmb{R} \in \mathscr{R}_j$, and define $\delta = \min_i \; R_i$. Hence, for all $\Omega \neq \emptyset$, we have $\delta  \leq \sum_{i \in \Omega} R_i / \vert \Omega \vert \leq  g_{j}(\Omega) / \vert \Omega \vert $. Therefore, $\min_{\Omega \neq \emptyset} g_{j}(\Omega) / \vert \Omega \vert  $ is an upper bound on $\min_i \; R_i$. Choosing $\pmb{R} = (\pi_0, ..., \pi_0) \in \mathscr{R}_j$, where $\pi_0 = \min_{\Omega \neq \emptyset } g_{j}(\Omega) / \vert \Omega \vert $, the upper bound is thus achieved.
\end{proof}
Thus, the maximum symmetric rate (which also maximizes the minimum rate due to Lemma \ref{lemma:fair}) at BS $j$ is
\begin{equation}\label{eq:26:2}
R_{\textrm{Sym}}^{ j} = \min_{\Omega_j \subseteq \left\lbrace 1, 2, ..., L \right\rbrace, \Omega_j \neq \emptyset} \frac{ g_{j}(\Omega_j) }{  \vert \Omega_j \vert } .
\end{equation}

Finally, to find the maximum symmetric rate network-wide one needs to compute $\min R_{\textrm{Sym}}^{ j}$ for $j \in \left\lbrace 1, 2, ..., L \right\rbrace $. In the following, we discuss how \eqref{eq:26:2} can be solved over various regions.

\textbf{\textit{SD}}: At BS $j$, the minimization over the SD region in \eqref{eq:16} can be carried out by solving
\begin{align}
&\min_{\Omega_j}  \quad  \; \frac{ I (  \hat{y}_{ji} ; \; \pmb{x}_{\Omega_j} \; \vert \; \pmb{x}_{\Omega_j^c} ) }{   \vert \Omega_j \vert } \label{eq:27}  \\
&\text{subject to} \quad  \Omega_j \subseteq \lbrace 1, 2, ..., L \rbrace.  \label{eq:27:1}
\end{align}
The inequality in \eqref{eq:17:6} then allows one to find an achievable lower bound to the above problem by solving
\begin{align}
[\mathcal{P}1] \quad\quad &\min_{\Omega_j} \quad  \dfrac{1}{ \vert \Omega_j \vert} \;  \log \left( 1 + \dfrac{ M \sum_{l \in \Omega_j} \sqrt{\rho_{\textrm{p}}} \rho_{\textrm{u}}   \beta_{jil} \alpha_{jil} }{   \sum_{l=1}^L \sum_{k=1}^K \rho_{\textrm{u}} \beta_{jkl} +1 } \right)  \label{eq:28}  \\
&\text{subject to} \quad  \Omega_j \subseteq \lbrace 1, 2, ..., L \rbrace.  \label{eq:28:1}
\end{align}

\textbf{\textit{SND}:} It can be seen from \eqref{eq:24} that the region achieved by SND at BS $j$ can not in general be represented by the intersection of a finite number of half-spaces and thus does not fall in the category of convex polytopes. Hence, finding the maximum symmetric rate over this region in general does not appear to have a closed-form formulation as in \eqref{eq:26:2}. However, in order to provide insights into the benefits of using SND, below we investigate S-SND which provides a lower bound to SND. The special cases of two-cell and three-cell systems, for which analyzing SND is tractable, are also studied at the end of this section.

\textbf{\textit{S-SND}:} Under S-SND, it is solved at BS $j$
\begin{align}
[\mathcal{P}2] \quad\quad &\min_{\Omega_j} \quad  \dfrac{1}{ \vert \Omega_j \vert} \;  \log \left( 1 + \dfrac{ M \sum_{l \in \Omega_j} \sqrt{\rho_{\textrm{p}}} \rho_{\textrm{u}}   \beta_{jil} \alpha_{jil} }{   \sum_{l=1}^L \sum_{k=1}^K \rho_{\textrm{u}} \beta_{jkl} +1 } \right)  \label{eq:29}  \\
&\text{subject to} \quad  \lbrace j \rbrace \subseteq \Omega_j \subseteq \lbrace 1, 2, ..., L \rbrace.  \label{eq:29:1}
\end{align}

Note that even though $[\mathcal{P}1]$ and $[\mathcal{P}2]$ have the same objective function, following the discussion below \eqref{eq:25} the solution $\Omega_j$ of $[\mathcal{P}2]$ must include the index $j$ associated with the rate $R_{ij}$, and is thus not necessarily identical to that of $[\mathcal{P}1]$. 

To tackle $[\mathcal{P}1]$ (or $[\mathcal{P}2]$), we first consider two extreme regimes of high and low SINR.
\subsection{High SINR regime}
In this regime, the values of $M$ and $L$ are such that
\begin{equation}\label{eq:30}
\log \left( 1 + \dfrac{ M \sum_{l \in \Omega_j} \sqrt{\rho_{\textrm{p}}} \rho_{\textrm{u}}   \beta_{jil} \alpha_{jil} }{   \sum_{l=1}^L \sum_{k=1}^K \rho_{\textrm{u}} \beta_{jkl} +1 } \right) \simeq \log (M).
\end{equation}
For instance, this approximation holds when the number of BS antennas $M$ is truly large but finite while the number of cells $L$ is fixed.

Thus, in this regime the minimization in both $[\mathcal{P}1]$ and $[\mathcal{P}2]$ is achieved by $\Omega_j^{{\ast}}= \lbrace 1, 2, ..., L \rbrace$, and thereby the maximum symmetric rate at BS $j$ is given by
\begin{align}\label{eq:29:2}
 R_{\textrm{Sym}}^{\textrm{SD}, j} = R_{\textrm{Sym}}^{\textrm{S-SND}, j} =  \; \; \frac{ I \left(  \hat{y}_{ji} ; \; x_1 [i], x_2 [i], ..., x_L [i] \; \right) }{ L},  
\end{align} 
which scales as $\mathcal{O} (\log M)$. As discussed before, the performance of SND is at least as good as SD and S-SND, i.e., $R_{\textrm{Sym}}^{\textrm{SND}, j} \geq R_{\textrm{Sym}}^{\textrm{SD}, j} = R_{\textrm{Sym}}^{\textrm{S-SND}, j}$. Thus, in the high SINR regime the maximum symmetric rate of SND occurs on one of it's region boundaries that scales as $\mathcal{O} (\log M)$. In other words, from \eqref{eq:24} the maximum symmetric rate achieved by SND in the high SINR regime belongs to the full MAC, i.e., $R_{\textrm{Sym}}^{\textrm{SND}, j} \in \mathcal{R}_{\textrm{MAC}(\lbrace 1, ..., L \rbrace,j)}^i$. Therefore, in the high SINR regime one can upper bound $R_{\textrm{Sym}}^{\textrm{SND}, j}$ by $R_{\textrm{Sym}}^{\textrm{SND}, j} \leq \dfrac{1}{L} \; I \left(  \hat{y}_{ji}  ; \; x_1 [i], x_2 [i], ..., x_L [i] \; \right)$. Consequently, we obtain for the high SINR regime $R_{\textrm{Sym}}^{\textrm{SND}, j} = R_{\textrm{Sym}}^{\textrm{SD}, j} = R_{\textrm{Sym}}^{\textrm{S-SND}, j}$. To find the allocation network-wide, denoted by $R_{\textrm{Sym}}$, one needs to calculate the smallest value of \eqref{eq:29:2} across all cells, i.e., 
\begin{align}
R_{\textrm{Sym}} =  \; \min_j \; \frac{ I \left(  \hat{y}_{ji}  ; \; x_1 [i], x_2 [i], ..., x_L [i] \; \right) }{   L} , \label{eq:29:3}
\end{align}
which is the same for all interference decoding schemes. Therefore, compared to TIN we obtain 
\begin{align}\label{eq:29:2:1}
R_{\textrm{Sym}} > \min_j I \left(  \hat{y}_{ji}  ; \; x_j [i] \; \right) = R_{\textrm{Sym}}^{\textrm{TIN}},
\end{align}
i.e., joint decoding of all signals $\lbrace x_l [i] \rbrace_{l=1}^L$ performs strictly better than decoding only the desired signal (e.g., $x_j [i]$ at BS $j$) while treating the interference signals (e.g., $\lbrace x_l [i] \rbrace_{l=1, l \neq j}^L$ at BS $j$) as noise (TIN). 
\begin{remark}\label{remark:ref}
\normalfont Similar to the results of \cite{adhikary2017uplink ,van2018large}, it is apparent from \eqref{eq:29:3} that, for sufficiently large $M$, the proposed interference decoding schemes of this paper are also able to achieve rates that scale as $\mathcal{O} ( \log M)$. Note that the achievable rates of \cite{adhikary2017uplink ,van2018large} are higher than the ones reported in this paper due to an extra array processing at a centralized network controller, which results in a larger pre-log factor (e.g., $1$ vs $1/L$ in the high SINR regime). However, such a centralized processing requires extra resources and hardware infrastructure to facilitate the BS cooperation at the network controller. In contrast, in this paper, all processing are performed locally at BSs without needing any cooperation. 
\end{remark}
\begin{remark}\label{remark:different:pilots}
\normalfont Consider the alternative approach of using different pilots in different cells, as explained before \eqref{eq:14:14}. One should note that, for the regime of large but finite $M$, decoding all $K (L-1)+1$ number of interfering users at the current BS will generally produce a smaller symmetric rate than the approach of Section \ref{sec:csi} which only decodes $L$ users, due to the much smaller pre-log factor in the former case. For instance, in the regime of high SINR, using \eqref{eq:30}-\eqref{eq:29:3}, the achieved maximum symmetric rate of the former case is $\approx 1/(K (L-1)+1) \log (M)$, whereas that of the latter case is $\approx 1/L \log (M)$. Hence, when decoding the PC interference, re-using orthonormal pilots cross all cells is preferred as, for finite $M$, it results in larger symmetric rate across the network.        
\end{remark}
\textbf{Observation:} In the high SINR regime, regardless of the cells geometry, all interference decoding schemes SD/SND/S-SND have identical performance and also strictly outperform TIN. In turn, one may choose to implement SD which has a simpler decoder.\vspace{-3mm}
\subsection{Low SINR regime} 
In this regime, the values of $M$ and $L$ are such that
\begin{equation}\label{eq:31}
\log \left( 1 + \dfrac{ M \sum_{l \in \Omega_j} \sqrt{\rho_{\textrm{p}}} \rho_{\textrm{u}}   \beta_{jil} \alpha_{jil} }{   \sum_{l=1}^L \sum_{k=1}^K \rho_{\textrm{u}} \beta_{jkl} +1 } \right) \simeq  \dfrac{ M \sum_{l \in \Omega_j} \sqrt{\rho_{\textrm{p}}} \rho_{\textrm{u}}   \beta_{jil} \alpha_{jil} }{   \sum_{l=1}^L \sum_{k=1}^K \rho_{\textrm{u}} \beta_{jkl} +1 } .
\end{equation}
For instance, this approximation holds when $M$ and $L$ are small such that the product of the number of BS antennas $M$ and the ratio $\frac{ \sum_{l \in \Omega_j} \sqrt{\rho_{\textrm{p}}} \rho_{\textrm{u}}   \beta_{jil} \alpha_{jil} }{   \sum_{l=1}^L \sum_{k=1}^K \rho_{\textrm{u}} \beta_{jkl} +1 }$ becomes small.

Since $\alpha_{jil} = \frac{\sqrt{\rho_{\rm p}} \beta_{jil}}{1+\rho_{\rm p} \sum_{l_1=1}^L \beta_{jil_1}}$, provided that \eqref{eq:31} holds one can see that $[\mathcal{P}1]$ (or $[\mathcal{P}2]$) has the same minimizer $\Omega_j^{\ast}$ as the problem below
\begin{align}
[\mathcal{P}3] \quad\quad \min_{\Omega_j}  \quad   \frac{ \sum_{l \in \Omega_j} \nolimits \beta_{jil}^2 }{ \vert \Omega_j \vert } ,  \label{eq:32}
\end{align}
subject to \eqref{eq:28:1}  for SD (or \eqref{eq:29:1} for S-SND). To solve this problem, we first construct the sorted vector
\begin{equation}\label{eq:32:1}
\pmb{\pi}_{ji} = \left[ \beta_{ji(1)}^2, \beta_{ji(2)}^2, ..., \beta_{ji(L)}^2 \right]^T,
\end{equation}
whose entries are $\beta_{ji1}, \beta_{ji2}, ..., \beta_{jiL}$ sorted in non-decreasing order, i.e., $\beta_{ji(1)}^2 \leq \beta_{ji(2)}^2 \leq ... \leq \beta_{ji(L)}^2$. Thus, considering a simple distance-based pathloss model for large-scale fading coefficients and assuming that users are associated to the nearest BS, $\beta_{ji(1)}$ and $\beta_{ji(L)}$ correspond to the furthest and closest users to BS $j$, respectively and $\beta_{ji(L)} = \beta_{jij}$. Further note that the objective function of $[\mathcal{P}3]$ is an averaging operation over a subset of the entries of \eqref{eq:32:1}. We now study the solution of $[\mathcal{P}3]$ under SD and S-SND in the following.

\textbf{\textit{SD}:} Subject to \eqref{eq:28:1}, the average over any subset of the entries of \eqref{eq:32:1} is always at least as large as the smallest entry, $\beta_{ji(1)}^2$, with equality when the average is only over the smallest entry. Therefore, $\Omega_j^{ {\ast}}$ is the index of the BS located at the furthest distance from the BS $j$. In other words, in the low SINR regime the performance of SD at BS $j$ is (unsurprisingly) limited by the rate of the furthest user from BS $j$, i.e., $R_{\textrm{Sym}}^{\textrm{SD}, j} = I (  \hat{y}_{ji} ; \; x_{\Omega_j^{\ast}} [i] \; \big\vert \; \pmb{x}_{S \setminus \Omega_j^{\ast}} )$

\textbf{\textit{S-SND}:} First, note that since $\Omega_j^{{\ast}}$ for S-SND at BS $j$ contains index $ j$, the averaging operation in $[\mathcal{P}3]$ must include the largest entry of \eqref{eq:32:1}, i.e., $\beta_{ji(L)}^2 = \beta_{jij}^2$. Moreover, by including any other entry of \eqref{eq:32:1} to the averaging operation, its value decreases. Therefore, to find $\Omega_j^{{\ast}}$ one should start with the initial value $b^1 = ( \beta_{ji(L)}^2 + \beta_{ji(1)}^2)/2$, and then repeatedly add the next smallest entry of the vector $\pmb{\pi}_{ji}$ to the averaging operation. This process terminates when the average becomes larger than its value from the previous iteration. The following algorithm computes $\Omega_j^{ {\ast}}$ for S-SND.

\begin{algorithm}\label{algorithm:2}
Set $q=1$, $b^q = ( \beta_{ji(L)}^2 + \beta_{ji(1)}^2)/2$.
\begin{description}
\item[{(1)}] Set $q=q+1$, and compute $b^q = \dfrac{\beta_{ji(L)}^2 + \sum_{l=1, l \neq L}^q \beta_{ji(l)}^2}{q+1}$. 
\item[{(2)}] If $ b^q \geq b^{q-1} $, then stop and output $\Omega_j^{{\ast}} = \lbrace j \rbrace \cup \left\lbrace l: \; \beta_{jil}^2 \in \lbrace \pmb{\pi}_{ji} [1:q-1] \rbrace \right\rbrace $, where $\pmb{\pi}_{ji} [1:q-1]$ denotes the the first $q-1$ entries of $\pmb{\pi}_{ji}$. Otherwise, go to step 1.
\end{description}
\end{algorithm}
Therefore, in the low SINR regime, by construction the allocation 
\begin{equation}\label{eq:33}
R_{\textrm{Sym}}^{\textrm{S-SND}, j} =  \quad \frac{ I (  \hat{y}_{ji} ; \; \pmb{x}_{\Omega_j^{ {\ast}}}  \; \vert \; \pmb{x}_{S \setminus \Omega_j^{ {\ast}}}  ) }{ \vert \Omega_j^{\ast} \vert} , 
\end{equation}
where $\Omega_j^{\ast} \subseteq S=\lbrace 1, 2, ..., L \rbrace$, is strictly larger than that of SD, i.e., in the low SINR regime S-SND strictly outperforms SD. Even though the performance of SND in the low SINR regime appears intractable, as previously pointed out S-SND provides a lower bound to SND. Hence, network-wide we have $R_{\textrm{Sym}}^{\textrm{SND}} \geq R_{\textrm{Sym}}^{\textrm{S-SND}} > R_{\textrm{Sym}}^{\textrm{SD}}$.
\subsection{General SINR}
Now, consider the problem of determining the maximum symmetric rate of SD ($[\mathcal{P}1]$) or S-SND ($[\mathcal{P}2]$) in general, where approximations of high and low SINR are no longer assumed. In Appendix~C, efficient methods to numerically compute the maximum symmetric rate of SD and S-SND are presented.

Since it is difficult to comment on the performance of maximum symmetric rate for SND in general due to the structure of the SND region, we next study two special cases of two-cell and three-cell systems which are analytically tractable. For the two-cell system, we find conditions under which either TIN/SND is optimal or interference decoding schemes SD/SND/S-SND are all optimal. Whereas, for the three-cell system we will briefly illustrate examples where SND outperforms \textit{all the other schemes}.
\subsubsection{Two-cell system}\label{subsection:two:cell}
We now consider a cellular system consisting of only two cells, and denote the indices of the cells by $j=1, 2$. Associated with the $i^{\rm th}$ user, $i=1, 2, ..., K$, the rate regions achieved at BS 1 are given as below.

\textbf{\textit{SD}:} From \eqref{eq:16} we obtain
\begin{align}
R_{i1} &\leq  I \left(  \hat{y}_{1i} ; \; x_1 [i] \; \big\vert \; x_2 [i] \right) \label{eq:37} \\ 
R_{i2} &\leq  I \left(  \hat{y}_{1i} ; \; x_2 [i] \; \big\vert \; x_1 [i] \right) \label{eq:38} \\
R_{i1} + R_{i2} &\leq  I \left(  \hat{y}_{1i} ; \; x_1 [i], \; x_2 [i] \; \right). \label{eq:39}
\end{align}
\textbf{\textit{SND}:} From \eqref{eq:24} we obtain
\begin{align}
R_{i1} &\leq  I \left(  \hat{y}_{1i} ; \; x_1 [i] \; \big\vert \; x_2 [i] \right) \label{eq:40} \\
R_{i1} + \min \left\lbrace R_{i2} , \; I \left(  \hat{y}_{1i} ; \; x_2 [i] \; \big\vert \; x_1 [i] \right) \right\rbrace &\leq  I \left(  \hat{y}_{1i} ; \; x_1 [i], \; x_2 [i] \; \right) \label{eq:41} .
\end{align}
\textbf{\textit{S-SND}:} From \eqref{eq:25} we obtain 
\begin{align}
R_{i1} &\leq  I \left(  \hat{y}_{1i} ; \; x_1 [i] \; \big\vert \; x_2 [i] \right) \label{eq:42} \\
R_{i1} + R_{i2} &\leq I \left(  \hat{y}_{1i} ; \; x_1 [i], \; x_2 [i] \; \right) \label{eq:43}.
\end{align}    
\begin{remark}\label{remark:bs}
\normalfont One can similarly obtain the rate regions at BS 2 by replacing $\hat{y}_{1i}$ with $\hat{y}_{2i}$ and swapping appropriate indices in \eqref{eq:37}-\eqref{eq:43}.
\end{remark}
An interesting observation for a two-cell system is that the region for SND is the union of the SD/S-SND region and the TIN region. We now aim to investigate the performance of different schemes with maximum symmetric allocation. For the two-cell system, we first define the following cases:

\textbf{\textit{Case (i)}:} In this case, we have
\begin{equation}\label{eq:48}
I \left(  \hat{y}_{1i} ; \; x_2 [i] \; \big\vert x_1 [i] \; \right) < I \left(  \hat{y}_{1i} ; \; x_1 [i] \; \right).
\end{equation}

\textbf{\textit{Case (ii)}:} In this case, we have
\begin{align}
\nonumber &\dfrac{1}{2} I \left(  \hat{y}_{1i} ; \; x_1 [i], \; x_2 [i] \; \right) \\
 &\hspace{2mm}\leq \min \left\lbrace I \left(  \hat{y}_{1i} ; \; x_1 [i] \; \big\vert \; x_2 [i] \right), \; I \left(  \hat{y}_{1i} ; \; x_2 [i] \; \big\vert \; x_1 [i] \right) \right\rbrace. \label{eq:44}
\end{align}

\textbf{\textit{Case (iii)}:} In this case, we have
\begin{equation}\label{eq:48:1}
I \left(  \hat{y}_{1i} ; \; x_1 [i] \; \big\vert x_2 [i] \; \right) < I \left(  \hat{y}_{1i} ; \; x_2 [i] \; \right).
\end{equation}

From the perspective of the maximum symmetric rate, cases (i)-(iii) refer to conditions (in terms of mutual information) under which the diagonal $R_{i2} = R_{i1}$ intersects a particular facet of the rate region. 
\begin{figure*}[t] 
	\centering 
	\captionsetup[subfloat]{captionskip=0mm}
	\scriptsize 
	\subfloat[]{\def\svgwidth{140pt}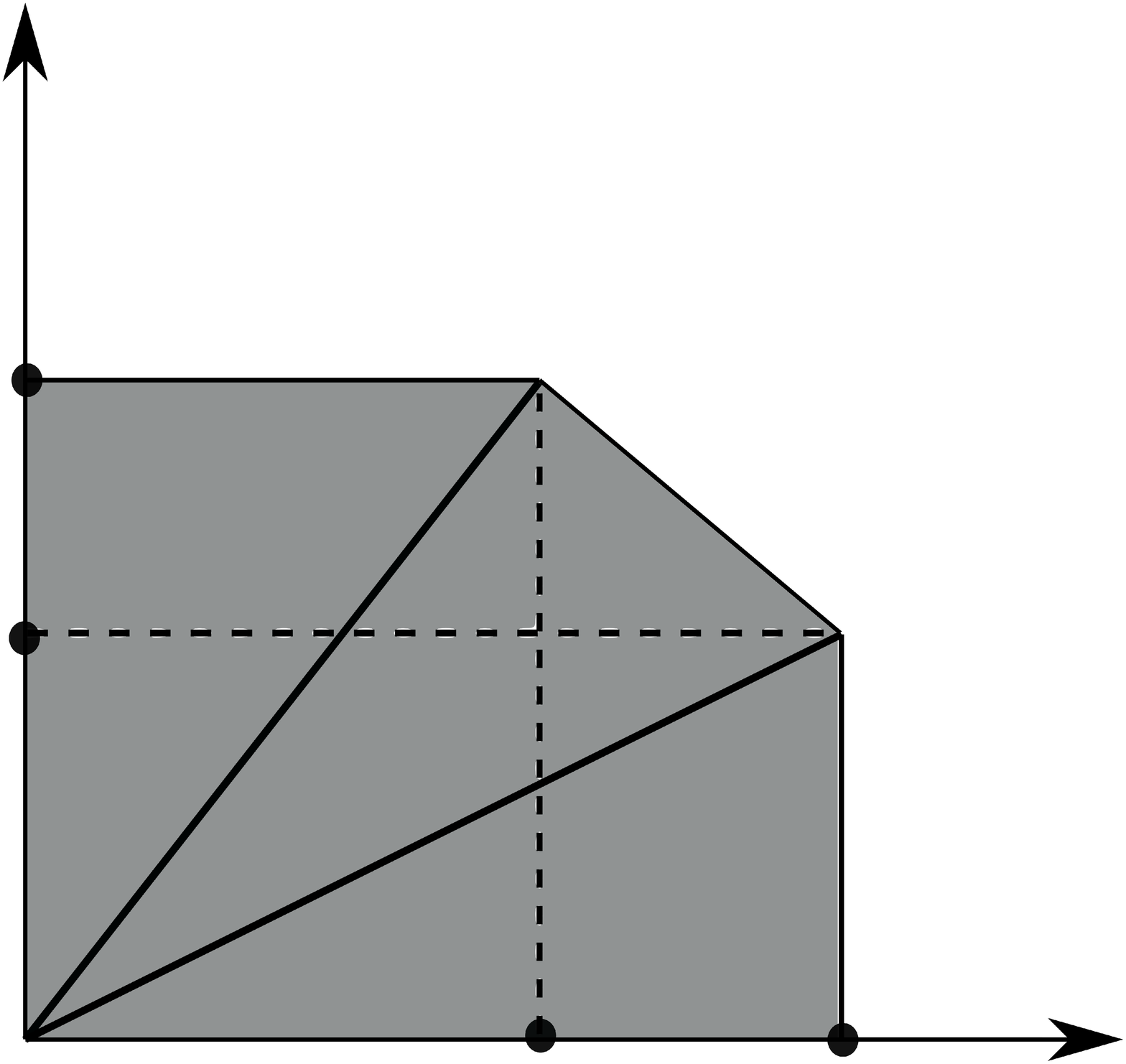\label{fig:cases:a}} \hspace{10mm}\vspace{0mm} 
	\subfloat[]{\def\svgwidth{125pt}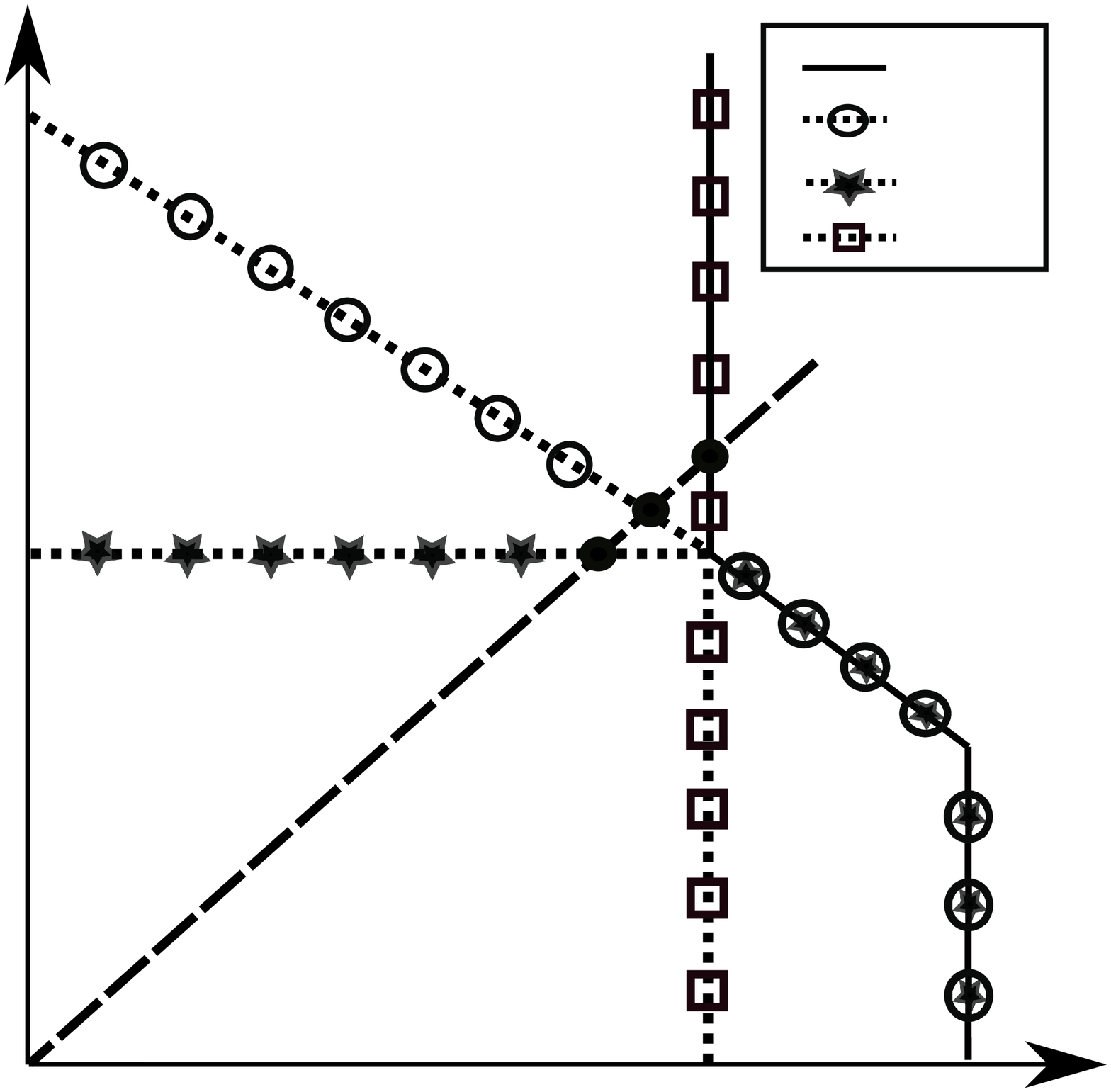\label{fig:cases:1}}\hspace{15mm}
	\subfloat[]{\def\svgwidth{120pt}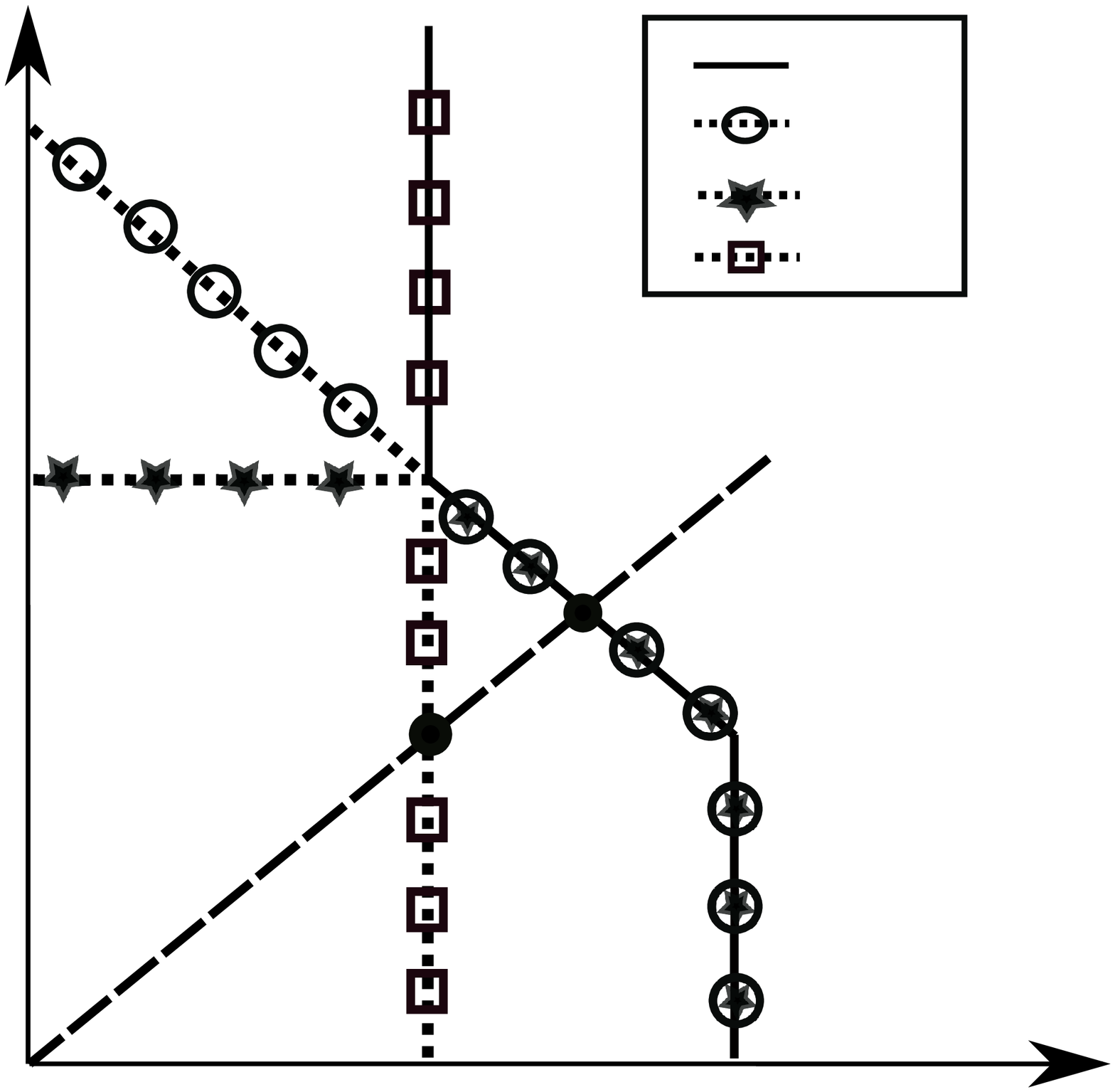\label{fig:cases:2}}
	\caption{\small (a) Illustration of the SD region at BS 1 representing the 3 sub-regions $G_1$, $G_2$ and $G_3$ over which the diagonal $R_{i2}=R_{i1}$ will intersect a particular facet of the rate region, (b) Illustration of the rate regions achieved under TIN/SND/S-SND/SD at BS 1 for case (i): the diagonal $R_{i2}=R_{i1}$ intersects SD at point E, S-SND at point F, and SND/TIN at point G, resulting in \eqref{eq:61:1}, (c) Illustration of the rate regions achieved under TIN/SND/S-SND/SD at BS 1 for case (ii): the diagonal $R_{i2}=R_{i1}$ intersects TIN at point H, and SND/S-SND/SD at point I, resulting in \eqref{eq:47:1}.\label{fig:cases}}
	\end{figure*}
	
 More specifically, consider the rate region achieved by SD at BS 1 depicted in Fig.~\ref{fig:cases:a}, where the entire region is divided into 3 sub-regions $G_1, G_2$ and $G_3$. Also, from \eqref{eq:37}-\eqref{eq:39}, note that the corner points are given by $(C, \; A) = \left( I \left(  \hat{y}_{1i} ; \; x_1 [i] \; \right), \; I \left(  \hat{y}_{1i} ; \; x_2 [i] \; \big\vert x_1 [i] \; \right) \right)$ and $(D, \; B) = \left( I \left(  \hat{y}_{1i} ; \; x_1 [i] \; \big\vert x_2 [i] \; \right), \; I \left(  \hat{y}_{1i} ; \; x_2 [i] \;  \right)  \right)$. Now, the conditions under which the diagonal $R_{i2} = R_{i1}$ lies in sub-regions $G_1, G_2$ or $G_3$, are equivalent to the conditions of the three cases of \eqref{eq:48}-\eqref{eq:48:1} as follows: the diagonal $R_{i2} = R_{i1}$ lies in $G_1$, \ied case (i) is true, iff $C > A$; the diagonal $R_{i2} = R_{i1}$ lies in $G_2$, \ied case (ii) is true, iff   $C \leq A$ and $B \leq D$; the diagonal $R_{i2} = R_{i1}$ lies in $G_3$, \ied case (iii) is true, iff  $B > D$.   Specifically, the conditions for case (i) in \eqref{eq:48} and case (iii) in \eqref{eq:48:1} are exactly those given by $C > A$ and $B > D$, respectively. 

For case (ii), note that one can also write 
 \begin{align}
I (  \hat{y}_{1i} ; \; x_1 [i], \; x_2 [i] \; ) &= I (  \hat{y}_{1i} ; \; x_1 [i] \; ) + I (  \hat{y}_{1i} ; \; x_2 [i] \; \vert \; x_1 [i] \; ) \label{eq:mutual:1} \\
&= I (  \hat{y}_{1i} ; \; x_2 [i] \; ) + I (  \hat{y}_{1i} ; \; x_1 [i] \; \vert \; x_2 [i] \; ). \label{eq:mutual:2}
\end{align}
Hence, in case (ii) where we have $C \leq A$ and $B \leq D$, by replacing $C$ and $B$ with their respective identity from \eqref{eq:mutual:1} and \eqref{eq:mutual:2}, we reach the following conditions 
\begin{align}
\frac{1}{2} I (  \hat{y}_{1i} ; \; x_1 [i], \; x_2 [i] \; ) &\leq I (  \hat{y}_{1i} ; \; x_2 [i] \; \vert \; x_1 [i] \; ) \label{eq:cond:1} \\
\frac{1}{2} I (  \hat{y}_{1i} ; \; x_1 [i], \; x_2 [i] \; ) &\leq I (  \hat{y}_{1i} ; \; x_1 [i] \; \vert \; x_2 [i] \; ) \label{eq:cond:2},
\end{align}
resulting in \eqref{eq:44}. 
\begin{remark}
\normalfont If the worst-case uncorrelated noise bound in \eqref{eq:17:6} is substituted for the mutual information expressions in \eqref{eq:48}-\eqref{eq:48:1}, case (iii) can never happen as the effects of small-scale fading vanish in \eqref{eq:17:6} and thus the received power of $x_2 [i]$ at BS 1 can not be larger than that of $x_1 [i]$ at BS 1. Hence, case (i) and case (ii) can be viewed as two complimentary and exhaustive conditions for a two-cell system at BS 1.

Note that the bounds of \eqref{eq:17:6} differ from the mutual expressions in \eqref{eq:48:1} due to two factors: (a) the expressions in \eqref{eq:48:1} depend on the specific fading gains, and (b) the effective noise is not necessarily Gaussian. However, in the limit of large $M$ the channel hardening of \eqref{eq:9} minimizes the effects of (a). Moreover, due to the channel hardening of \eqref{eq:9} as well as the assumption of Gaussian signaling in Theorem \ref{lemma:2}, the interference terms (effective noise) in \eqref{eq:8:1} are asymptotically Gaussian.

\end{remark}
The performance comparison of various schemes at BS 1 is summarized in the following corollary. 
\begin{corollary}\label{cor:1}
If the condition of case (i) in \eqref{eq:48} holds at BS 1, then
\begin{equation}\label{eq:61:1} 
R_{\textrm{Sym}}^{\textrm{SD},1} <  R_{\textrm{Sym}}^{\textrm{S-SND},1} < R_{\textrm{Sym}}^{\textrm{SND},1} = R_{\textrm{Sym}}^{\textrm{TIN},1},
\end{equation}
otherwise, if the condition of case (ii) in \eqref{eq:44} holds at BS 1, then
\begin{equation}\label{eq:47:1}
R_{\textrm{Sym}}^{\textrm{TIN},1} \leq R_{\textrm{Sym}}^{\textrm{SD},1} = R_{\textrm{Sym}}^{\textrm{SND},1} = R_{\textrm{Sym}}^{\textrm{S-SND},1},
\end{equation}
with strict equality in \eqref{eq:47:1} if and only if \eqref{eq:44} holds with strict equality.
\end{corollary}
\begin{proof}
See Appendix D.
\end{proof}
Fig. \ref{fig:cases} illustrates an example of this corollary. Sub-figure (a) represents case (i) and its consequence in \eqref{eq:61:1}, whereas sub-figure (b) represents case (ii) and its consequence in \eqref{eq:47:1}.

To comment on the performance of various schemes over both cells, we consider a symmetric setting which is easy to analyze, and provides insights into the benefits of employing interference decoding schemes.  

We define the symmetric setting as a scenario, where the MACs at both BS 1 and 2 are identical. Therefore, if case (i) is active at BS 1, it is also active at BS 2, and the resulting rates are equal at both BSs. Following Remark \ref{remark:bs} it is thus obtained network-wide that
\begin{equation}\label{eq:47:1:1}
R_{\textrm{Sym}}^{\textrm{SD}} <  R_{\textrm{Sym}}^{\textrm{S-SND}} < R_{\textrm{Sym}}^{\textrm{SND}} = R_{\textrm{Sym}}^{\textrm{TIN}}.
\end{equation}
\textbf{Observation:} Both SND and TIN achieve the same performance and strictly outperform SD and S-SND. Thus, TIN may be the better choice of strategy in practice due to its simplicity.

Similarly, if case (ii) is active with strict inequality at BS 1, it is also active with strict inequality at BS 2, and the resulting rates are equal at both BSs. Following Remark \ref{remark:bs} it is thus obtained network-wide that
\begin{equation}\label{eq:47:1:2}
R_{\textrm{Sym}}^{\textrm{TIN}} < R_{\textrm{Sym}}^{\textrm{SD}} = R_{\textrm{Sym}}^{\textrm{SND}} = R_{\textrm{Sym}}^{\textrm{S-SND}}.
\end{equation}
\textbf{Observation:} The interference decoding schemes SD/SND/S-SND achieve the same performance and strictly outperform TIN. Thus, SD (joint decoding of both users) may be the simplest one to implement in practice.
Practical examples of these cases will be demonstrated in the numerical results section. 

Consider, for instance, a setup where all users are located at a distance $x$ from the corresponding BS as in Fig. \ref{fig:two:circle}. With respect to the lower bound in \eqref{eq:17:6}, this setup is symmetric as the effects of small-scale fading vanish in \eqref{eq:17:6}.
\begin{figure}[t!]
\centering
\def\svgwidth{260pt} 
\footnotesize
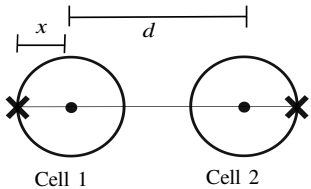
\caption{\small An example of symmetric geometry with circular cells at a fixed distance of $d$ from each other, where users are located at a distance $x$ from the their BSs located at the center of the cells. The position of users is denoted by '$\times$'.}\label{fig:two:circle}
\end{figure} 
Note that in a two-cell system, irrespective of whether the setting is symmetric or not, under no circumstances does SND strictly outperform all the other schemes. We next illustrate scenarios for a three-cell system, where SND can strictly outperform \textit{all the other schemes}.

\subsubsection{Three-cell system}\label{subsection:three:cell}
Now, consider a cellular system consisting of only three cells, where the indices of the cells are denoted by $j=1, 2, 3$. In this case, the rate regions under SD/S-SND can be obtained by a straightforward extension of \eqref{eq:37}-\eqref{eq:39} and \eqref{eq:42}-\eqref{eq:43} to the three-cell system, thus omitted for brevity. Moreover for SND, the rate region associated with the $i^{\rm th}$ user, $i=1, 2, ..., K$, at BS 1 can be found using \eqref{eq:24} as follows:
\begin{align}\label{eq:62}
R_{i1} &\leq  I \left(  \hat{y}_{1i} ; \; x_1 [i] \; \big\vert \; x_2 [i], x_3 [i]   \; \right) \\
\nonumber R_{i1} &+ \min  \left\lbrace I \left(  \hat{y}_{1i} ; \; x_2 [i] \; \big\vert \; x_1 [i], x_3 [i]   \; \right), \; R_{i2} \right\rbrace \\
 &\leq  I \left(  \hat{y}_{1i} ; \; x_1 [i], x_2 [i] \; \big\vert \; x_3 [i] \; \right) \label{eq:63}\\
\nonumber R_{i1} &+ \min  \left\lbrace I \left(  \hat{y}_{1i} ; \; x_3 [i] \; \big\vert \; x_1 [i], x_2 [i]   \; \right), \; R_{i3} \right\rbrace \\
  &\leq  I \left(  \hat{y}_{1i} ; \; x_1 [i], x_3 [i] \; \big\vert \; x_2 [i] \; \right)  \label{eq:64}\\
\nonumber R_{i1} &+ \min  \big\lbrace I \left(  \hat{y}_{1i} ; \; x_2 [i], x_3 [i] \; \big\vert \; x_1 [i] \; \right), \\
 \nonumber &\hspace{11mm}R_{i2} + I \left(  \hat{y}_{1i} ; \; x_3 [i] \; \big\vert \; x_1 [i], x_2 [i]   \; \right),  \\
\nonumber &\hspace{11mm}R_{i3} + I \left(  \hat{y}_{1i} ; \; x_2 [i] \; \big\vert \; x_1 [i], x_3 [i]   \; \right), \\
\nonumber &\hspace{11mm}R_{i2} + R_{i3} \big\rbrace  \\
 &\leq  I \left(  \hat{y}_{1i} ; \; x_1 [i], x_2 [i], x_3 [i]   \; \right). \label{eq:65}
\end{align}

An example of this region is plotted in Fig. \ref{fig:3}, where the dashed lines indicate that the region at BS 1 is unbounded in variables $R_{i2}$ and $R_{i3} $, which is in agreement with property [P1] of the achievable region. Also, note that following Remark \ref{remark:bs} the regions corresponding to BSs 2 and 3 can be similarly found.
\begin{figure}[t!]
\centering
\includegraphics[width=7cm]{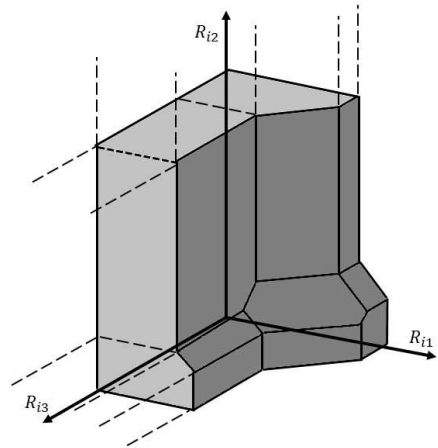}
\caption{\small An example of the rate region obtained by SND at BS 1.\label{fig:3}}
\end{figure} 
By comparing \eqref{eq:62}-\eqref{eq:65} with the achievable regions of SD and S-SND, it is noted that there are four faces in Fig. \ref{fig:3} that are only achieved by SND and not by any other schemes. More precisely at BS 1, it is possible for $R_{\rm Sym}^1$ to achieve one of the rates, $I (  \hat{y}_{1i} ; \; x_1 [i] \; \big\vert \;  x_3 [i]   \; )$, $I (  \hat{y}_{1i} ; \; x_1 [i] \; \big\vert \; x_2 [i]  \; )$, $ \frac{1}{2} I (  \hat{y}_{1i} ; \; x_1 [i] , x_2 [i] \; )$ or $ \frac{1}{2} I (  \hat{y}_{1i} ; \; x_1 [i] , x_3 [i] \; )$. Note that the first rate $I (  \hat{y}_{1i} ; \; x_1 [i] \; \big\vert \;  x_3 [i]   \; )$ can be interpreted as the maximum rate of the $i^{\rm th}$ user of cell 1, while treating the $i^{\rm th}$ user in cell 2 as noise. The second rate $I (  \hat{y}_{1i} ; \; x_1 [i] \; \big\vert \; x_2 [i]  \; )$ can be interpreted similarly. Moreover, the rate $ \frac{1}{2} I (  \hat{y}_{1i} ; \; x_1 [i] , x_2 [i] \; )$ can be interpreted as the maximum symmetric rate achieved by joint decoding of the $i^{\rm th}$ users of cells 1 and 2, while treating the $i^{\rm th}$ user of cell 3 as noise. The fourth rate $ \frac{1}{2} I (  \hat{y}_{1i} ; \; x_1 [i] , x_3 [i] \; )$ can be interpreted similarly. Therefore, neither SD/S-SND nor TIN can provide these rates, in which case it is conceivable that SND could strictly outperform \textit{all the other schemes}. More discussion will be provided in the numerical results section.
\section{Numerical Results}
To illustrate the performance of the two-cell system under both cases of (i) and (ii), we consider two different scenarios. In scenario (a), we assume that the cell radius and the distance of BSs are fixed, while the number of antennas $M$ varies. In scenario (b), we assume that $M$ and the distance of BSs are fixed, while the cell radius varies. In both scenarios for the geometry setting shown in Fig. \ref{fig:two:circle}, we compare the performance of various schemes TIN/SD/SND based on their achieved maximum symmetric rate using the bounds in \eqref{eq:17:6}. In particular, for scenario (a), we numerically quantify a threshold on $M$ at which the transition from case (i) to case (ii) is observed. Analogously, for scenario (b), we numerically quantify a threshold on the distance $x$ at which the transition from case (i) to case (ii) is observed. It should be pointed out that the results of this section are numerical examples presented only for the sake of illustration that validate the analytical findings of the previous sections. Hence, the identified thresholds depend on specific choices of system parameters. Moreover, for the large-scale fading coefficients $\beta_{jil}$ we use a distance-based path loss model similar to \cite{hoydis2013massive} and neglect shadowing, i.e, $\beta_{jil}=( d_0/d_{jil})^{\alpha}$, where $d_{jil}$ is the distance of the $i^{\rm th}$ user in cell $l$ from the BS $j$, $\alpha$ is the path loss exponent and $d_0$ is a normalization constant. For both scenarios (a) and (b), we take $K=4$, $\alpha = 2$, $\rho_{\rm u} = 30$, $\rho_{\rm p} = 120$, and $d_0=100$ m.
\begin{figure*}[t!] 
	\captionsetup[subfloat]{captionskip=0mm}
	\centering    
	\subfloat[Scenario (a)]{\includegraphics[width=8.5cm, height=6.5 cm]{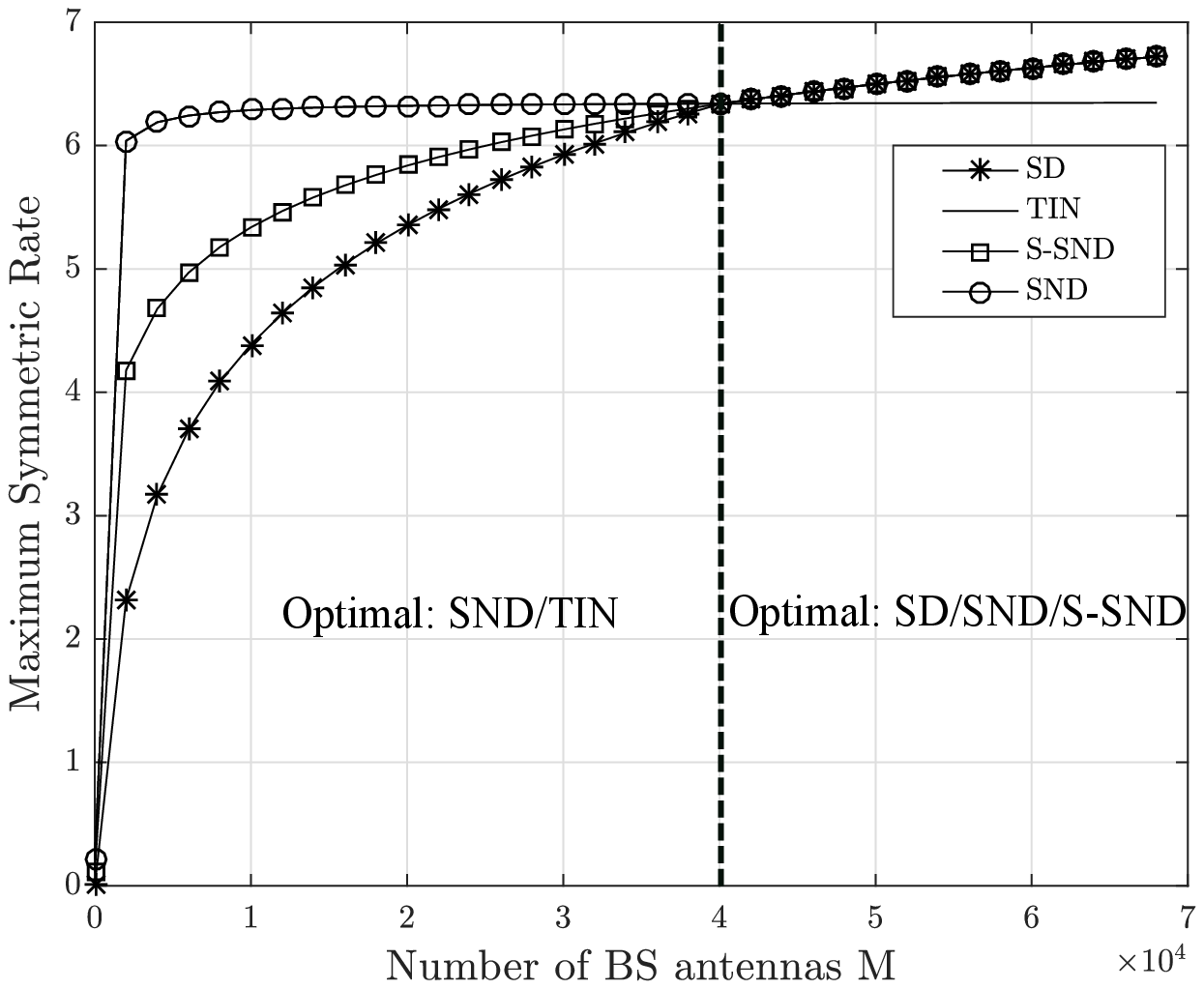}\label{fig:two-cell:rate}} %
	\subfloat[Scenario (b)]{\includegraphics[width=8.5cm, height=6.5 cm]{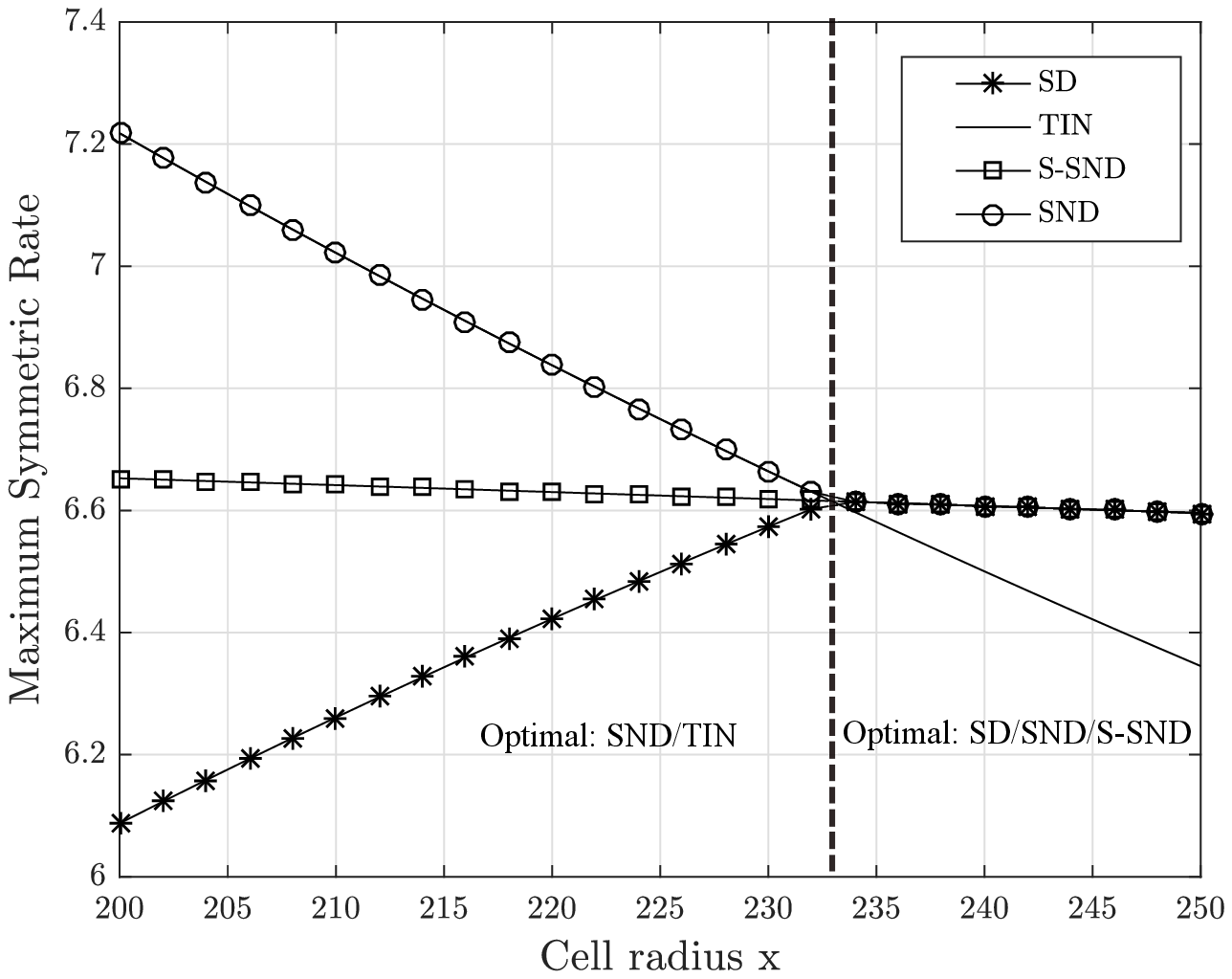}\label{fig:two-cell-x:rate}}\hspace{5cm}
\caption{\small Achieved maximum symmetric rate across two cells, $R_{\textrm{Sym}}$, over regions of TIN/SD/SND/S-SND: (a) for fixed values of $d = 2x$ and $x=400$ m, as $M$ is increased; (b) for fixed values of $d = 500$ m and $M =  5 \! \times \! 10^4$, as $x$ is increased.}
\end{figure*}
Achieved maximum symmetric rate of various schemes for scenario (a) are shown in Fig. \ref{fig:two-cell:rate}, where $x= 400$ m and $d= 2 x$. It can be observed that for $M < 4 \times 10^4$, condition of case (i) in \eqref{eq:48} is active; thus, SND and TIN have the same performance and strictly outperform SD/S-SND, i.e.,  $R_{\textrm{Sym}}^{\textrm{SD}} <  R_{\textrm{Sym}}^{\textrm{S-SND}} < R_{\textrm{Sym}}^{\textrm{SND}} = R_{\textrm{Sym}}^{\textrm{TIN}} $. In other words, for $M < 4 \times 10^4$, to achieve the optimum performance each BS should only decode the signal of its own user while treating the signal of PC interference as noise. On the other hand, when $M > 4 \times  10^4$, the condition of case (ii) in \eqref{eq:44} is active; thus, interference decoding schemes are all optimal, i.e., $R_{\textrm{Sym}}^{\textrm{TIN}} < R_{\textrm{Sym}}^{\textrm{SD}} =  R_{\textrm{Sym}}^{\textrm{S-SND}} = R_{\textrm{Sym}}^{\textrm{SND}} $. Consequently, for significantly large values of $M$, to achieve the optimum performance each BS should jointly decode both the signal of its own user as well as that of the PC interference. This observation also matches with the consequence of the high SINR regime for truly large $M$ in \eqref{eq:29:2:1}. Also, notice that there does not exist any range of $M$ for which SND is strictly optimal. These observations are all in agreement with the analysis performed in subsection \ref{subsection:two:cell}. 
 
Next consider scenario (b) where the BSs are at a distance of $d=500$ m, $M = 5 \times 10^4$, and the cell radius $x$ varies from $200$ m to $250$ m.  The maximum symmetric rate for various schemes in this scenario, shown in Fig. \ref{fig:two-cell-x:rate}, illustrates that when approximately $x < 233$ m, the condition of case (i) in \eqref{eq:48} is active, and thus using TIN is optimal. One implication of this observation is that for a fixed $M$, there exists a threshold on cell radius such that if $x$ is smaller than this threshold (i.e., $ 233$ m in this example), interfering users located in the other cell are far away from the BS of the current cell, and hence treating interfering users as noise is optimal. On the other hand, if $x$ is above the threshold, the condition of case (ii) in \eqref{eq:44} is active: the interfering users are now close to the current BS, and thus interference decoding schemes achieve the optimal performance.   
\begin{figure}[t!]
\scriptsize
\centering 
\def\svgwidth{160pt} 
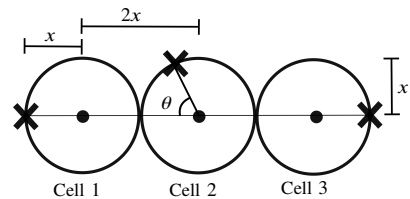 
\caption{\small Illustration of a three-cell system with circular cells, where all users of the left and the right cells are located on the cell edge at the farthest distance from the BSs located at the center of the cells, whereas the position of users on the edge of the middle cell is changing over $0^{\circ} \leq \theta \leq 180^{\circ}$. The position of users is denoted by '$\times$'.} 
\label{fig:three:circle}
\end{figure}
\begin{figure}[t!]
\centering
\epsfig{width=9cm, figure=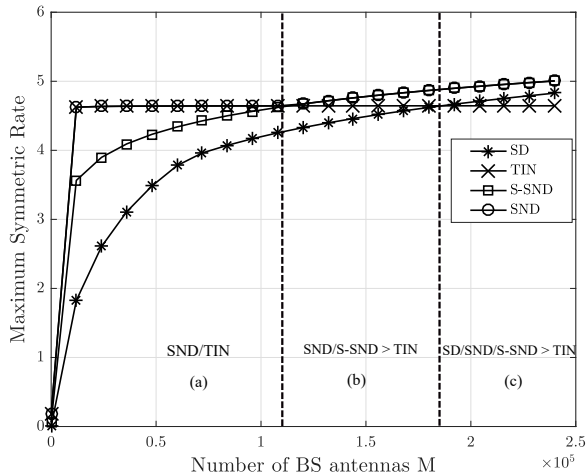}
\caption{\small Achieved maximum symmetric rate across three cells, $R_{\textrm{Sym}}$, over all four regions TIN/SD/SND/S-SND. The x-axis is divided into three intervals: (a) SND and TIN have identical performance and strictly outperform SD/S-SND, (b) SND/S-SND strictly outperform TIN/SD, (c) SD/SND/S-SND strictly outperform TIN.}
\label{fig:three-cell:rate}
\end{figure} 
\begin{figure*}[t] 
	\captionsetup[subfloat]{captionskip=0mm}
	\centering    
	\subfloat[$M=10^3$]{\includegraphics[width=7cm, height=6 cm]{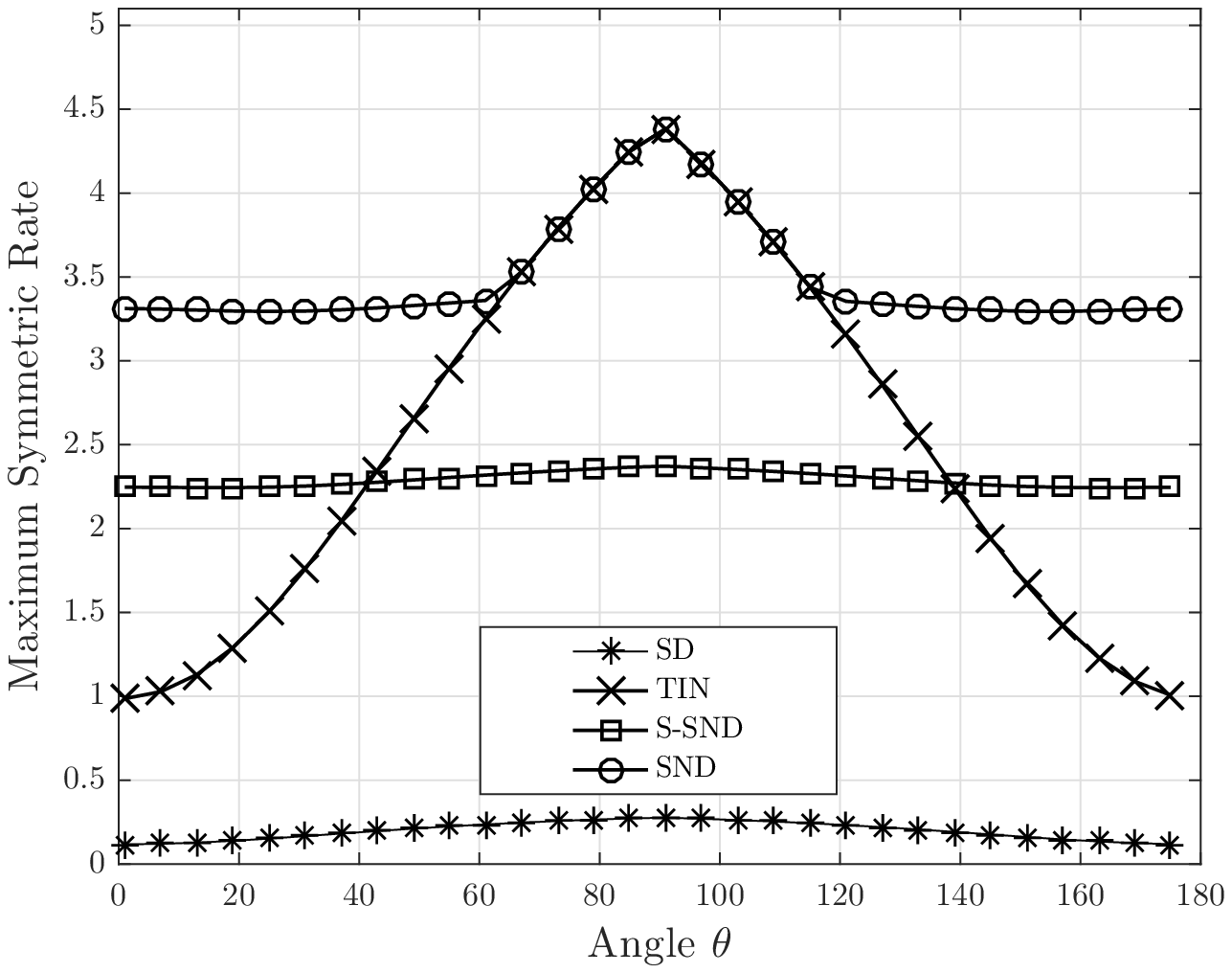}\label{fig:theta:1}} %
	\subfloat[$M=10^4$]{\includegraphics[width=7cm, height=6 cm]{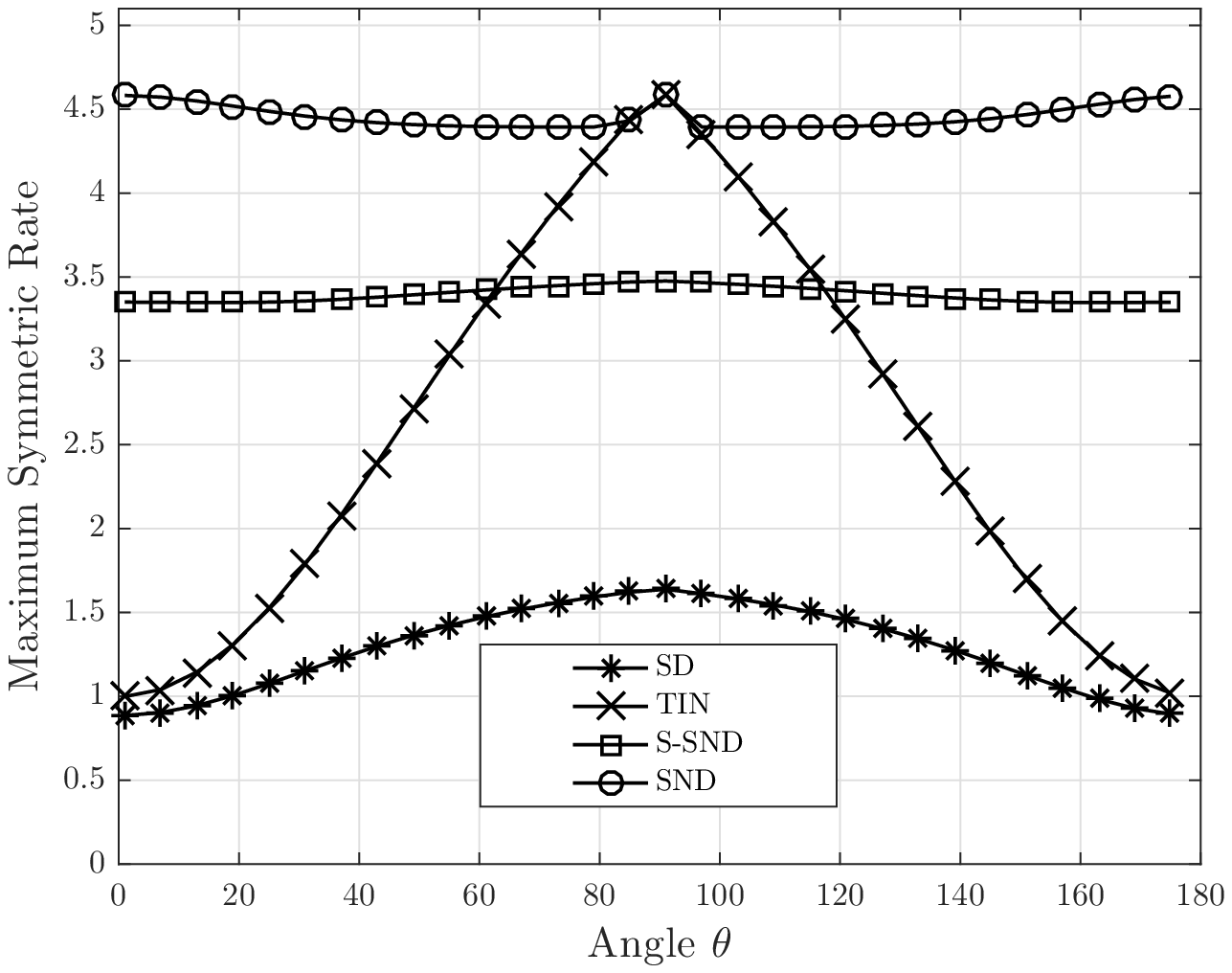}\label{fig:theta:2}}\hspace{5cm}
	\subfloat[$M=5 \times 10^4$]{\includegraphics[width=7cm,height=6 cm]{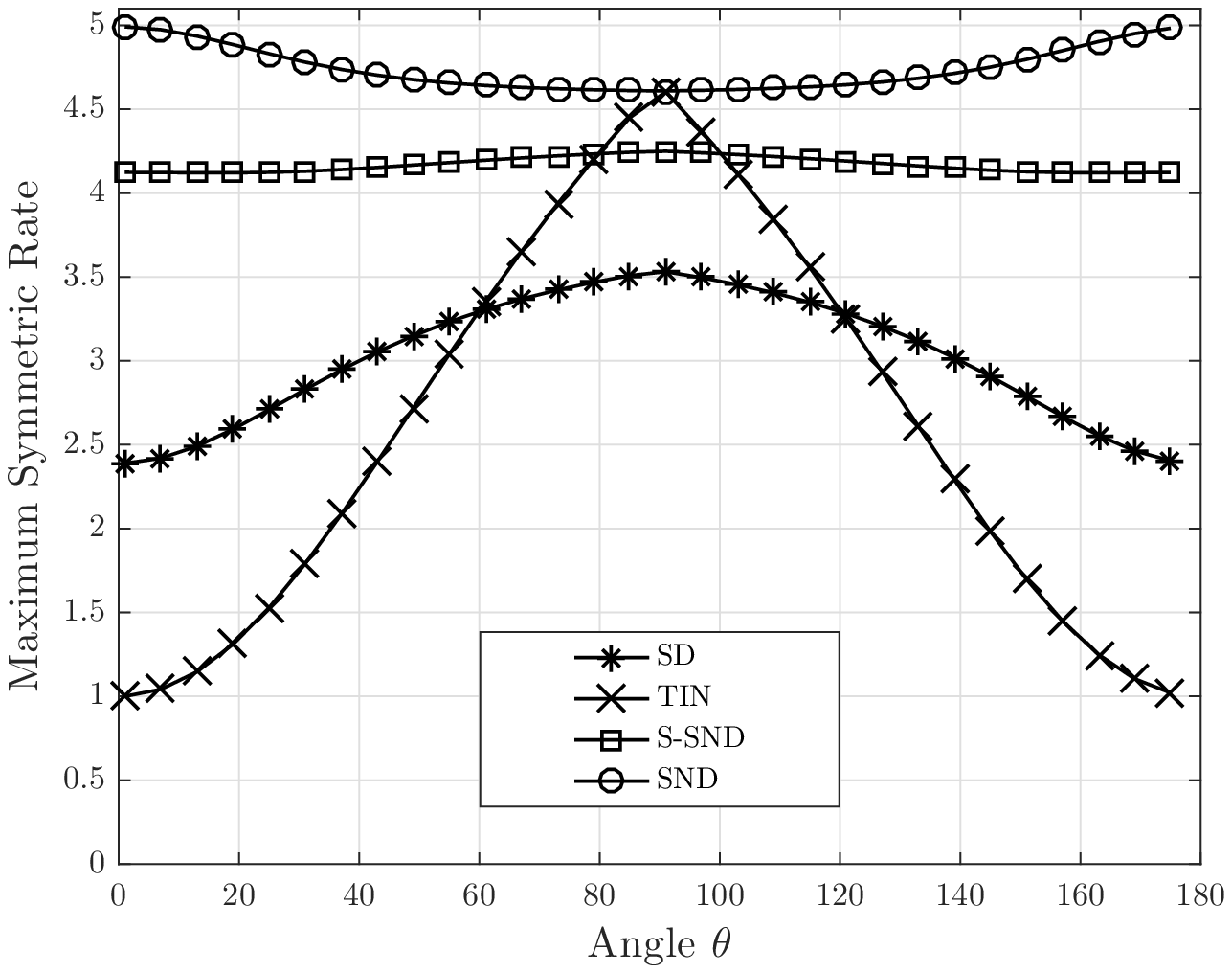}\label{fig:theta:3}} 
	\subfloat[$M= 10^5$]{\includegraphics[width=7cm,height=6 cm]{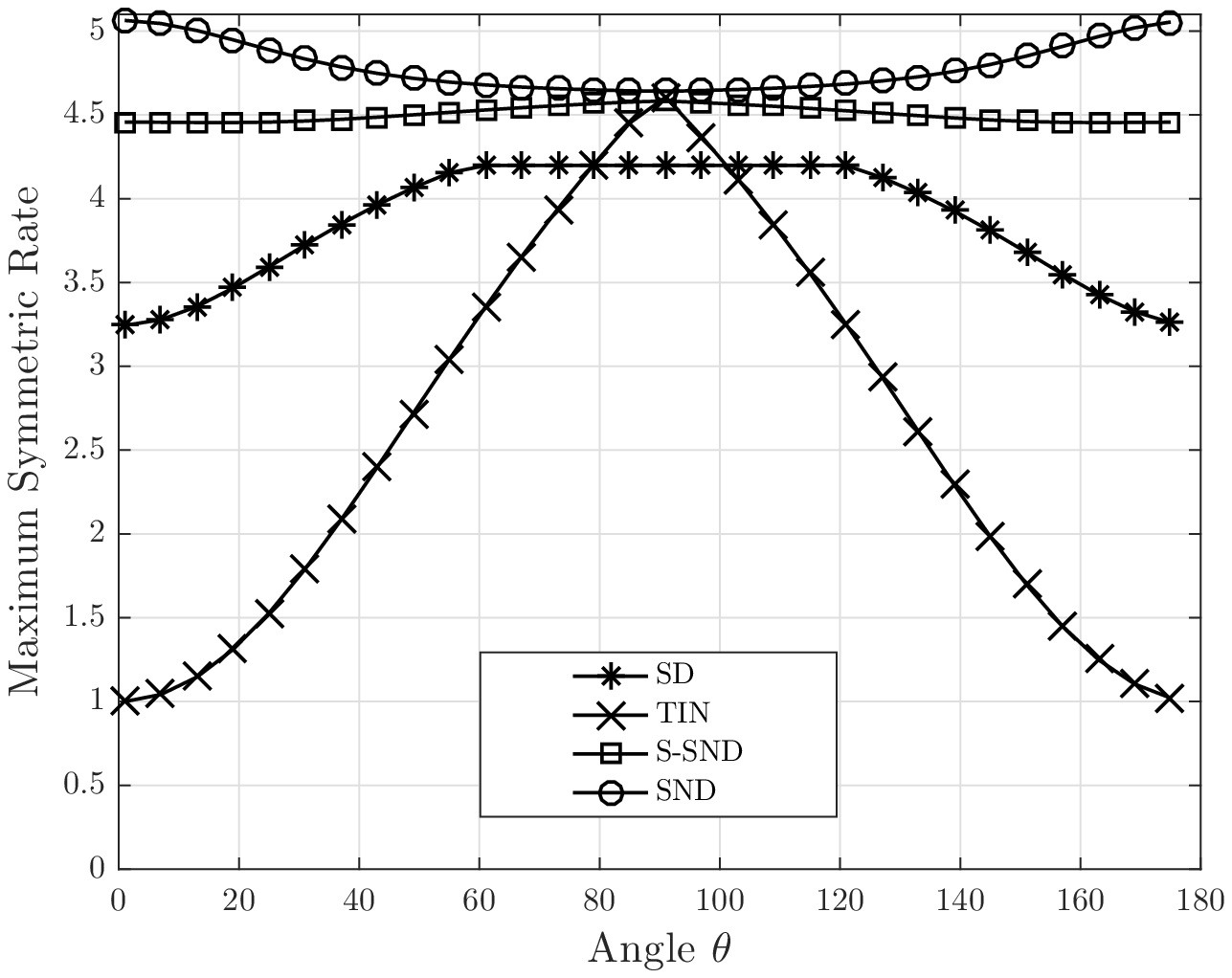}\label{fig:theta:4}}%
	\caption{\small Achieved maximum symmetric rates across three cells, $R_{\textrm{Sym}}$, over all four regions TIN/SD/SND/S-SND for $0^{\circ} \leq \theta \leq 180^{\circ}$, (a) $M=10^3$, (b) $M=10^4$, (c) $M=5 \times 10^4$, (d) $M=10^5$. \label{fig:theta}}
\end{figure*}
We now extend the cell configuration model considered in Fig. \ref{fig:two:circle}, to the case of three cells. Particularly, we consider two scenarios with three circular cells based on Fig. \ref{fig:three:circle} with $x = 400$ m, where the position of users in the left and the right cells is fixed on the cell edge at the maximum distance from the BSs located at the center, whereas those of the middle cell are determined based on the angle $\theta$. More specifically, in scenario (a), it is assumed that the users of the middle cell are located on the cell edge at $\theta=90^{\circ}$, while in scenario (b), it is assumed that the position of users on the edge of the middle cell are swept over $0^{\circ} \leq \theta \leq 180^{\circ}$. Note that while the setting of Fig. \ref{fig:three:circle} is impractical as all users will not be located at a single point at cell edge (i.e, at farthest distance from their BSs) in practice, nevertheless it does provide a conservative and somewhat pessimistic estimate of the user rates. In practice, users will typically be distributed more uniformly in the cell and some users will thus be located closer to their BSs. Hence, rate for users located closer to the BS will be higher than those located at cell edge, and thus the overall rate will be potentially higher.

The maximum symmetric rates of scenario (a) for various schemes and for different values of $M$ are illustrated in Fig. \ref{fig:three-cell:rate}, where the parameters of the setup are the same as before. It is observed from Fig. \ref{fig:three-cell:rate} that, even though for $M < 1.1 \times 10^5$ the performance of SND and TIN are identical and strictly better than SD/S-SND (and one can thus simply use TIN), for $M > 1.1 \times 10^5$ the PC interference is ``strong" in that decoding it, as opposed to treating it as noise, produces better rates, i.e., SND/S-SND strictly outperform \textit{all the other schemes} (and one should thus only use SND). In other words, for $1.1 \times 10^5 < M < 5 \times 10^5$ (the upper bound on $M$ is not shown in the figure), the optimum performance is achieved only by SND and not by any other scheme. However, for $M > 5 \times 10^5$ all interference decoding schemes achieve the same rate as expected, while outperforming TIN. Finally, observe that for (approximately) $M > 1.85 \times 10^5 $ SD outperforms TIN as well.

Finally, while keeping all other parameters the same as before, we consider scenario (b) where the position of users in the middle cell changes with $\theta$, and plot the resulting maximum symmetric rates for various schemes against $0^{\circ} \leq \theta \leq 180^{\circ}$  with different values of BS antennas as illustrated in Fig. \ref{fig:theta}: (a) $M=10^3$, (b) $M=10^4$, (c) $M=5 \times 10^4$, and (d) $M=10^5$. We focus on the rates for $0^{\circ} \leq \theta \leq 180^{\circ}$ as, due to symmetry, the rates for $180^{\circ} < \theta < 360^{\circ}$ are identical to those for $0^{\circ} \leq \theta \leq 180^{\circ}$.  Also, recall that the SND region, which contains TIN, S-SND and SD regions as special cases, as explained in subsection III-B, is always optimal. 

  Observe that in Fig. \ref{fig:theta:1} for $M=10^3$, the characteristics of rates can be classified into 3 regimes of $\theta$:  regime-1 where $\theta$ is smaller than $90^{\circ}$ and SND outperforms TIN (i.e., $\theta \leq 66^\circ$ in the setting of  Fig. \ref{fig:theta:1}), regime-2 where $\theta$ is close to $90^{\circ}$ and SND and TIN have the same performance (i.e., $66^\circ < \theta \leq 114^\circ$), and regime-3 where $\theta$  is larger than $90^{\circ}$ and SND outperforms TIN again (i.e., $\theta \geq 114^\circ$). In regime-1, we have $R_{\rm Sym}^{\rm SND} > R_{\rm Sym}^{\rm S-SND} > R_{\rm Sym}^{\rm TIN} > R_{\rm Sym}^{\rm SD}$, as explained below. Note that $\theta$ captures the distance between users and BSs in different cells, and for $\theta$ in regime-1, users in the middle cell are much closer to BS~1 and farther away from BS~3. Therefore, users in the middle cell creates ``strong'' PC interference at BS~1 and ``weak'' PC interference at BS~3; hence, SND/S-SND outperforms all other schemes as it allows for users of the middle cell to be decoded at BS~1 and to be treated as noise at BS~3. In comparison, TIN provides poor performance as treating ``strong'' users from the middle cell as noise at BS~1 drastically reduces the rates, whereas SD performs poorly as it requires decoding of ``weak'' PC interference from users in the middle cell at BS~3. In the complementary setting of regime-3, the same principles apply with the roles of BS~1 exchanged with BS~3.
    
   In contrast, for $\theta $ in regime-2, we have $R_{\rm Sym}^{\rm SND} = R_{\rm Sym}^{\rm TIN} > R_{\rm Sym}^{\rm S-SND} > R_{\rm Sym}^{\rm SD}$. Here, users in the middle cell are somewhat far from both BS~1 and BS~2, and thus the resulting PC interference becomes ``weak" at both BS~1 and BS~2. Hence, performing TIN at both BSs is optimal and provides identical performance to that of SND while outperforming S-SND and SD.
   
   As $M$ is increased, the decoding rates continue to follow the same trend, but regime-2 shrinks to a small set of angles near $90^\circ$, while the two other regimes expand, as illustrated in Fig.~\ref{fig:theta:1} to Fig.~\ref{fig:theta:4}. As $M$ increases, users in the middle cell produce progressively stronger interference at BS~1 and BS~3, and thus using TIN at BS~1 and BS~3 to treat these users as noise results in poor rates except for a small set of $\theta$ near $90^\circ$. These trends are similar to those in Fig. \ref{fig:three-cell:rate}, where $\theta$ was fixed at $\theta = 90^\circ$.  These trends also confirm that when considering more than two cells, depending on the geometry and the parameters of the setting (e.g., $M$, cell radius, etc), there exist scenarios where SND performs strictly better than TIN/SD/S-SND.  
\section{Conclusion}
In this paper, to address the PC problem in the uplink of a multi-cell massive MIMO system it was proposed to decode the interference caused by PC rather than treating it as noise. In particular, when MRC is used at BS, it was shown that by coding over multiple coherence intervals and decoding the PC interference, the per-user rates tend to infinity as $M \rightarrow \infty$. Moreover, it was shown that when decoding the interference, using the same pilots across all cells (as opposed to using different pilots) is preferred as it results in decoding significantly fewer interference terms at each BS. A worst-case uncorrelated noise technique was also established for multiple access channels, from which achievable rates under two interference decoding schemes SD/SND were found for finite $M$. Comparing the performance of different schemes based on their maximum symmetric rate, structural results were found for the extreme regimes of high and low SINR. Specifically, for the high SINR regime when $M$ is truly large, it was found that all interference decoding schemes achieve the maximum sum-rate and thus strictly outperform TIN.

The special cases of two-cell and three-cell systems were also studied. In the case of a two-cell system with symmetric geometry, conditions were found under which either SND and TIN achieve the same performance and are optimal (one should thus treat interference as noise), or all interference decoding schemes achieve the same performance and are optimal (one should thus jointly decode the desired signal as well as the PC interference). Furthermore, the analytical findings were numerically validated by quantifying a threshold on $M$ (or on the cell radius), where TIN/SND was shown to be optimal below this threshold. On the other hand, beyond this threshold it was observed that only the interference decoding schemes achieve optimum performance. 

Also in the case of a three-cell system, it was numerically shown that there exists a range of $M$ for which the optimum performance is achieved only by SND and not by any other scheme. Hence, it was concluded that for large enough $M$ when there are more than two cells (which is indeed true in practice), SND can strictly outperform \textit{all the other schemes}.  

One possible future extension is to consider the downlink counterpart of this problem using well-known linear precoding techniques such as maximum ratio transmission (MRT), zero forcing (ZF), etc., where each BS simultaneously serves $K$ users inside its cell. Specifically, after performing an arbitrary precoding technique at all BSs, $K$ non-interfering $L$-user ICs will be obtained, whereby with simultaneous unique/non-unique decoding (SD/SND) of the intended signal along with the PC interference at each user, one can find achievable rates similar to \eqref{eq:17:6} that scale as $\mathcal{O} (\log M)$. Another possible future extension is to consider a correlated Rayleigh fading channel that will lead to a non-diagonal channel covariance matrix. This change of channel model will change the MMSE estimate of the channel vector in \eqref{eq:6}-\eqref{eq:6:1} and consequently the distribution of the channel estimate and the estimation error, and thus the power of different terms in $P_1$-$P_4$ in \eqref{eq:17:1}-\eqref{eq:17:4}. One should note that, even though the updated expressions of \eqref{eq:17:1}-\eqref{eq:17:4} result in a new rate lower bound, it would still grow as $\mathcal{O} \log (M)$ (similar to \eqref{eq:17:6}), and thereby the final conclusions will remain the same.

 



%

\appendices
\section{}

Without loss of generality assume that $\Omega= \left\lbrace 1, 2, ..., l \right\rbrace$ and thus $\Omega^c = \left\lbrace l+1, ..., L \right\rbrace$. We start by expanding the r.h.s of \eqref{eq:w} as follows
\begin{align}\label{eq:w1}
I \left( \pmb{x}_{\Omega}^G ; \; y \big\vert \pmb{x}_{\Omega^c}^G  \right) &= h \left( \pmb{x}_{\Omega}^G \big\vert \pmb{x}_{\Omega^c}^G \right) - h \left( \pmb{x}_{\Omega}^G \big\vert y, \pmb{x}_{\Omega^c}^G \right) \\
&\stackrel{(a)}{=} h \left( \pmb{x}_{\Omega}^G \right) - h \left( \pmb{x}_{\Omega}^G \big\vert y, \pmb{x}_{\Omega^c}^G \right) \\
&\stackrel{(b)}{=} \log \left( (\pi e)^l \Pi_{i=1}^l P_i \right) - h \left( \pmb{x}_{\Omega}^G \big\vert y, \pmb{x}_{\Omega^c}^G \right), \label{eq:w11}
\end{align}
where $(a)$ is because the entries of $\pmb{x}_{\Omega}^G$ and $\pmb{x}_{\Omega^c}^G$ are independent, and $(b)$ follows from the entropy of a complex Gaussian vector with independent entries. Also, using the chain rule one can write  
\begin{align} \label{eq:w2}
&h \left( \pmb{x}_{\Omega}^G \big\vert y, \pmb{x}_{\Omega^c}^G \right) \\
\nonumber &= \sum_{i \in \Omega} h \left( x_i^G \Big\vert x_1^G, ..., x_{i-1}^G, y, \pmb{x}_{\Omega^c}^G \right) \\
\nonumber &= \sum_{i \in \Omega} h \left( x_i^G - \alpha_i \left( y - \sum_{j=1}^{i-1} x_j^G - \sum_{k \in \Omega^c} x_k^G \right) \Big\vert x_1^G, ..., x_{i-1}^G, y, \pmb{x}_{\Omega^c}^G \right),  
\end{align}
where $\alpha_i$ is any constant. Defining $\tilde{y}_i = y - \sum_{j=1}^{i-1} x_j^G - \sum_{k \in \Omega^c} x_k^G $, we obtain
\begin{align}\label{eq:w3}
\sum_{i \in \Omega} &\hspace{1mm}h \left( x_i^G - \alpha_i \tilde{y}_i \big\vert x_1^G, ..., x_{i-1}^G, y, \pmb{x}_{\Omega^c}^G \right) \\ 
\nonumber &\stackrel{(c)}{\leq} \sum_{i \in \Omega} h \left( x_i^G - \alpha_i \tilde{y}_i  \right) \\
 &\stackrel{(d)}{\leq} \sum_{i \in \Omega} \log \left( (\pi e) \text{var} \left[ x_i^G - \alpha_i \tilde{y}_i \right] \right),
\end{align}
where $(c)$ is due to the fact that conditioning reduces the entropy and $(d)$ follows as Gaussian distributions maximize entropy. To obtain the tightest upper bound, one should minimize $\text{var} [ x_i^G - \alpha_i \tilde{y}_i ]$, i.e., $\alpha_i \tilde{y}_i$ must be the LMMSE estimate of $x_i^G$. More precisely, one can choose $\alpha_i = \mathbb{E} [ \tilde{y}_i^{\ast} \tilde{y}_i ]^{-1} \mathbb{E} [ x_i^{G} \tilde{y}_i^{\ast} ] = P_i / (\sum_{j=i}^l P_j + \sigma_z^2)$, where the second equality follows since $z$ is uncorrelated from the users' signals. Thus
\begin{equation}\label{eq:w6}
\text{var} \left[ x_i^G - \alpha_i \tilde{y}_i \right] = \dfrac{ P_i \left( \sum_{j=i+1}^l P_j + \sigma_z^2 \right) }{ \sum_{j=i}^l P_j + \sigma_z^2 },
\end{equation}
and therefore we obtain
\begin{align}\label{eq:w7}
h \left( \pmb{x}_{\Omega}^G \big\vert y, \pmb{x}_{\Omega^c}^G \right) &\leq \sum_{i \in \Omega}  \log  \left( (\pi e) \dfrac{P_i \left( \sum_{j=i+1}^l P_j + \sigma_z^2 \right)}{\sum_{j=i}^l P_j + \sigma_z^2} \right) \\
&= \log \left( (\pi e)^l \dfrac{\left( \Pi_{i \in \Omega} P_i \right) \sigma_z^2}{ \sum_{i \in \Omega} P_i + \sigma_z^2} \right) .
\end{align}
Hence, from \eqref{eq:w11} the following lower bound is obtained
\begin{align}
\nonumber I &\left( \pmb{x}_{\Omega}^G ; \; y \big\vert \pmb{x}_{\Omega^c}^G  \right) \\
&\geq \log \left( (\pi e)^l \Pi_{i \in \Omega} P_i \right) - \log \left( (\pi e)^l \dfrac{\left( \Pi_{i \in \Omega} P_i \right) \sigma_z^2}{ \sum_{i \in \Omega} P_i + \sigma_z^2} \right) \label{eq:w8} \\
&= \log \left( 1 + \dfrac{\sum_{i \in \Omega} P_i}{\sigma_z^2} \right)  = I \left( \pmb{x}_{\Omega}^G ; \; y^G \Big\vert \pmb{x}_{\Omega^c}^G \right).
\end{align}

\section{}
We start by computing the power of the desired signals, $P_1$. Note that
\begin{align}\label{eq:app:1}
P_1 &= \sum_{l \in \Omega} \rho_{\rm u} \left\vert \mathbb{E} \left[ \pmb{\hat{g}}_{jij}^{\dag} \pmb{g}_{jil} \right] \right\vert^2 = \sum_{l \in \Omega} \rho_{\rm u} \left\vert \mathbb{E} \left[ \pmb{\hat{g}}_{jij}^{\dag} \left( \pmb{\hat{g}}_{jil} + \pmb{e}_{jil} \right) \right] \right\vert^2  \\
&\stackrel{(a)}{=} \sum_{l \in \Omega} \rho_{\rm u} \left\vert \mathbb{E} \left[ \pmb{\hat{g}}_{jij}^{\dag}  \pmb{\hat{g}}_{jil}  \right] \right\vert^2  \\
&=  \sum_{l \in \Omega} M^2 \rho_{\rm u} \left( \dfrac{\beta_{jil}}{\beta_{jij}} \right)^2 \dfrac{\rho_{\rm p}^2 \beta_{jij}^4}{\left( 1 + \rho_{\rm p} \sum_{l_1=1}^L \beta_{jil_1} \right)^2} \\
  &= M^2 \sum_{l \in \Omega} \rho_{\textrm{p}} \rho_{\textrm{u}}   \beta_{jil}^2 \alpha_{jij}^2,
\end{align}
where (a) follows from the fact that $\pmb{\hat{g}}_{jij}$ and $\pmb{\hat{e}}_{jij}$ are independent. Note that as explained earlier, all terms in the effective noise are uncorrelated; thus $ \textrm{var} [ z_{jij}^{\prime} ] = P_2 + P_3 + P_4.$

For the power of interference due to the channel estimation error, $P_2$, we write
\begin{align}\label{eq:app:3}
P_2 &= \sum_{l=1}^{L} \rho_{\textrm{u}} \mathbb{E} \left[ \left\vert  \pmb{\hat{g}}_{jij}^{\dag} \pmb{g}_{jil} - \mathbb{E} \left[ \pmb{\hat{g}}_{jij}^{\dag} \pmb{g}_{jil} \right] \right\vert^2  \right] \\
&= \sum_{l=1}^{L} \rho_{\textrm{u}} \mathbb{E} \left[ \left\vert  \pmb{\hat{g}}_{jij}^{\dag} \pmb{\hat{g}}_{jil} - \mathbb{E} \left[ \pmb{\hat{g}}_{jij}^{\dag} \pmb{\hat{g}}_{jil} \right] \right\vert^2  \right] \hspace{-1mm}+ \hspace{-1mm}\sum_{l=1}^L \rho_{\rm u}  \mathbb{E} \left[ \left\vert \pmb{\hat{g}}_{jij}^{\dag}  \pmb{e}_{jil}  \right\vert^2 \right] . \label{eq:app:3:1}
\end{align} 
For the first term in \eqref{eq:app:3:1} we obtain
\begin{align}
\nonumber \sum_{l=1}^{L} &\rho_{\textrm{u}} \mathbb{E} \left[ \left\vert  \pmb{\hat{g}}_{jij}^{\dag} \pmb{\hat{g}}_{jil} - \mathbb{E} \left[ \pmb{\hat{g}}_{jij}^{\dag} \pmb{\hat{g}}_{jil} \right] \right\vert^2  \right] \\
 &= \sum_{l=1}^L \rho_{\rm u} \left( \dfrac{\beta_{jil}}{\beta_{jij}} \right)^2  \textrm{var}  \left[ \pmb{\hat{g}}_{jij}^{\dag} \pmb{\hat{g}}_{jij} \right] \\
 &=  \sum_{l=1}^L  M \rho_{\rm u} \left( \dfrac{\beta_{jil}}{\beta_{jij}} \right)^2   \dfrac{\rho_{\rm p}^2 \beta_{jij}^4}{\left( 1 + \rho_{\rm p} \sum_{l_1=1}^L \beta_{jil_1} \right)^2}  .\label{eq:app:4}
\end{align}
 Similarly, for the second term in \eqref{eq:app:3:1} we obtain
\begin{align}
\nonumber \sum_{l=1}^L &\rho_{\rm u}  \mathbb{E} \left[ \left\vert \pmb{\hat{g}}_{jij}^{\dag}  \pmb{e}_{jil}  \right\vert^2 \right] \\
 &= \sum_{l=1}^L \rho_{\rm u}  \mathbb{E} \left[ \textrm{tr} \left(  \pmb{\hat{g}}_{jij}^{\dag}  \pmb{e}_{jil} \pmb{e}_{jil}^{\dagger} \pmb{\hat{g}}_{jij} \right) \right] \\
&= M \rho_{\rm u}  \left( \dfrac{\rho_{\rm p} \beta_{jij}^2}{ 1 + \rho_{\rm p}  \sum_{l_1=1}^L \beta_{jil_1}} \right)  \sum_{l=1}^L  \left( \beta_{jil} -  \dfrac{\rho_{\rm p} \beta_{jil}^2}{1 + \rho_{\rm p} \hspace{-1mm}\sum_{l_1 =1}^L \beta_{jil_1}} \right).\label{eq:app:5}
\end{align}
Therefore, using \eqref{eq:app:4} and \eqref{eq:app:5}, one can verify that
\begin{align}\label{eq:app:6}
P_2 &=  \sum_{l=1}^L  M \rho_{\rm u} \left( \dfrac{\beta_{jil}}{\beta_{jij}} \right)^2  \left(  \dfrac{\rho_{\rm p}^2 \beta_{jij}^4}{\left( 1 + \rho_{\rm p} \sum_{l_1=1}^L \beta_{jil_1} \right)^2} \right)  \\
\nonumber &\hspace{4mm}+  M \rho_{\rm u} \left(  \dfrac{\rho_{\rm p} \beta_{jij}^2}{ 1 + \rho_{\rm p} \sum_{l_1=1}^L \beta_{jil_1}} \right) \hspace{-1mm}\sum_{l=1}^L  \left( \beta_{jil} -  \dfrac{\rho_{\rm p} \beta_{jil}^2}{1 + \rho_{\rm p} \hspace{-1mm}\sum_{l_1 =1}^L \beta_{jil_1}} \right)  \\
&=  M  \sqrt{\rho_{\textrm{p}}} \beta_{jij} \alpha_{jij} \sum_{l=1}^L  \rho_{\textrm{u}} \beta_{jil} .
\end{align}
For the power of the interference of other users, $P_3$, we write
\begin{align}\label{eq:app:7}
P_3 &= \sum_{l=1}^L \sum_{k=1,k \neq i}^K \rho_{\textrm{u}}  \mathbb{E}  \left[  \left\vert \pmb{\hat{g}}_{jij}^{\dag} \pmb{g}_{jkl}  \right\vert^2  \right] \\
&= \sum_{l=1}^L \sum_{k=1,k \neq i}^K \rho_{\textrm{u}}  \mathbb{E}  \left[  \left\vert \pmb{\hat{g}}_{jij}^{\dag} \pmb{\hat{g}}_{jkl}  \right\vert^2  \right]  + \sum_{l=1}^L \sum_{k=1,k \neq i}^K \hspace{-2mm}\rho_{\textrm{u}} \mathbb{E}  \left[  \left\vert \pmb{\hat{g}}_{jij}^{\dag} \pmb{e}_{jkl}  \right\vert^2  \right]  \\
\nonumber &= \sum_{l=1}^L \sum_{k=1,k \neq i}^K M \rho_{\textrm{u}}  \left( \dfrac{\rho_{\rm p} \beta_{jij}^2}{1+ \rho_{\rm p} \sum_{l_1=1}^L  \beta_{jil_1} } \right) \left(  \dfrac{\rho_{\rm p} \beta_{jkl}^2 }{1 + \rho_{\rm p} \sum_{l_2=1}^L \beta_{jkl_2} } \right)  \\
\nonumber &\hspace{5mm}+  \sum_{l=1}^L \sum_{k=1,k \neq i}^K M \rho_{\textrm{u}}  \left( \dfrac{\rho_{\rm p} \beta_{jij}^2 }{1 + \rho_{\rm p} \sum_{l_1=1}^L \beta_{jil_1}} \right)  \\
 &\hspace{23mm}\times \left( \beta_{jkl} -  \dfrac{\rho_{\rm p} \beta_{jkl}^2}{1 + \rho_{\rm p} \sum_{l_2 =1}^L \beta_{jkl_2}} \right) \\
&=  M  \sqrt{\rho_{\textrm{p}}} \beta_{jij} \alpha_{jij} \sum_{l=1}^L \sum_{k=1, k \neq i}^K  \rho_{\textrm{u}}  \beta_{jkl} .
\end{align}
Finally, for the power of the noise, $P_4$, we obtain
\begin{align}\label{eq:app:8}
P_4 &= \mathbb{E}  \left[ \left\vert \pmb{\hat{g}}_{jij}^{\dag} \pmb{n}_j \right\vert^2 \right] = \textrm{tr} \left( \mathbb{E} \left[ \pmb{\hat{g}}_{jij} \pmb{\hat{g}}_{jij}^{\dag}  \right] \mathbb{E} \left[ \pmb{n}_j \pmb{n}_j^{\dagger} \right] \right) \\
&=  M \dfrac{\rho_{\rm p} \beta_{jij}^2}{1 + \rho_{\rm p} \sum_{l_1=1}^L \beta_{jil_1}} =  M \sqrt{\rho_{\textrm{p}}} \beta_{jij} \alpha_{jij} .
\end{align}

\section{}
We now propose efficient methods to compute the maximum symmetric rate of SD ($[\mathcal{P}1]$) or S-SND ($[\mathcal{P}2]$), where the assumption of high and low SINR are no longer made.

\textbf{\textit{SD}:} One may first compare all $2^L -1$ (total number of non-empty subsets of $S = \left\lbrace 1, 2, ..., L \right\rbrace$) different rates with a computational complexity that is exponential in the number of cells. However, due to the structure of $[\mathcal{P}1]$ we can find the minimum by comparing only $L$ rates at each BS. Thus network-wide, this task can be done in $\mathcal{O} (L^2)$ time as there are in total $L$ cells in the network. For notational brevity define $s_q = \sum_{l=1}^q \beta_{ji(l)}^2$ as the sum of the first $q$ entries of $\pmb{\pi}_{ji}$ and also
\begin{equation}\label{eq:34}
\mu_{ji} =  \dfrac{ M  \rho_{\textrm{p}} \rho_{\textrm{u}} }{  \left( \sum_{l=1}^L \sum_{k=1}^K \rho_{\textrm{u}} \beta_{jkl} +1 \right) \left( 1 + \rho_{\rm p} \sum_{l=1}^L \beta_{jil} \right)}.
\end{equation}
One can see that due to the structure of the entries of $\pmb{\pi}_{ji}$, if $\vert \Omega_j^{{\ast}} \vert = q$ for some integer $1 \leq q \leq L$, the minimum of the objective function in $[\mathcal{P}1]$ is $1/q \; \log \left( 1 + \mu_{ji} s_q \right)$. Hence, define
\begin{equation}\label{eq:35}
v_q =  \log \left( 1 + \mu_{ji} s_q \right) \; / \; q, \quad \forall  q  .
\end{equation}
Then, $[\mathcal{P}1]$ reduces to $\min_q v_q$, which can be calculated by comparing $L$ values.

\textbf{\textit{S-SND}:} Similarly when S-SND is used at BS $j$, define
\begin{equation}
c_q =
\begin{cases} 
\log \left( 1 + \mu_{ji} \beta_{jij}^2 \right), & \quad  q =1  \label{eq:36} \\
\dfrac{1}{q} \log \left( 1 + \mu_{ji} \left( \beta_{jij}^2 + \sum_{l=1, l \neq j}^q \beta_{ji(l)}^2 \right) \right), & \quad \textrm{otherwise} .
\end{cases}
\end{equation}
Therefore, $[\mathcal{P}2]$ reduces to $\min_q c_q$, which can also be calculated by comparing $L$ values.

\section{}
Using \eqref{eq:37}-\eqref{eq:43}, it can be verified that in case (i) we have $ R_{\rm Sym}^{\textrm{SD},1} = I (  \hat{y}_{1i} ; \; x_2 [i] \; \vert \; x_1 [i] \; )$, $R_{\textrm{Sym}}^{\textrm{S-SND},1} = \frac{1}{2} I (  \hat{y}_{1i} ; \; x_1 [i], \; x_2 [i] \; )$, and $R_{\textrm{Sym}}^{\textrm{SND},1} = R_{\textrm{Sym}}^{\textrm{TIN},1} = I (  \hat{y}_{1i} ; \; x_1 [i] \; )$.

Furthermore, using \eqref{eq:mutual:1}, one can rewrite \eqref{eq:48} as $I (  \hat{y}_{1i} ; \; x_2 [i] \; \vert \; x_1 [i] \; ) < \frac{1}{2} I (  \hat{y}_{1i} ; \; x_1 [i], \; x_2 [i] \; )$, and conclude that $R_{\textrm{Sym}}^{\textrm{SD},1} <  R_{\textrm{Sym}}^{\textrm{S-SND},1}$. Moreover, using \eqref{eq:mutual:1}, one can rewrite \eqref{eq:48} as $\frac{1}{2} I (  \hat{y}_{1i} ; \; x_1 [i], \; x_2 [i] \; ) < I (  \hat{y}_{1i} ; \; x_1 [i] \; )$, and conclude that $R_{\textrm{Sym}}^{\textrm{S-SND},1} < R_{\textrm{Sym}}^{\textrm{SND},1} = R_{\textrm{Sym}}^{\textrm{TIN},1}$. This is illustrated in Fig. \ref{fig:cases:1}. Similarly, using \eqref{eq:37}-\eqref{eq:43}, it can be verified that in case (ii) we have $R_{\textrm{Sym}}^{\textrm{TIN},1} = I (  \hat{y}_{1i} ; \; x_1 [i] \; )$, and $R_{\textrm{Sym}}^{\textrm{SD},1} = R_{\textrm{Sym}}^{\textrm{SND},1} = R_{\textrm{Sym}}^{\textrm{S-SND},1} = \frac{1}{2} I (  \hat{y}_{1i} ; \; x_1 [i], \; x_2 [i] \; )$. Also, using \eqref{eq:mutual:1}-\eqref{eq:mutual:2}, one can rewrite \eqref{eq:44} as \newline $ \max \lbrace I (  \hat{y}_{1i} ; \; x_1 [i]  ), \; I (  \hat{y}_{1i} ; \; x_2 [i]  ) \rbrace  \leq \frac{1}{2} I (  \hat{y}_{1i} ; \; x_1 [i], \; x_2 [i] \; ) $. Hence, when this condition holds it yields \eqref{eq:47:1}, which is also illustrated in Fig. \ref{fig:cases:2}.

\bibliography{IEEEabrv,RefTCOM}

\end{document}

%% file: Cells.eps_tex
\begingroup%
  \makeatletter%
  \providecommand\color[2][]{%
    \errmessage{(Inkscape) Color is used for the text in Inkscape, but the package 'color.sty' is not loaded}%
    \renewcommand\color[2][]{}%
  }%
  \providecommand\transparent[1]{%
    \errmessage{(Inkscape) Transparency is used (non-zero) for the text in Inkscape, but the package 'transparent.sty' is not loaded}%
    \renewcommand\transparent[1]{}%
  }%
  \providecommand\rotatebox[2]{#2}%
  \ifx\svgwidth\undefined%
    \setlength{\unitlength}{1909.41725363bp}%
    \ifx\svgscale\undefined%
      \relax%
    \else%
      \setlength{\unitlength}{\unitlength * \real{\svgscale}}%
    \fi%
  \else%
    \setlength{\unitlength}{\svgwidth}%
  \fi%
  \global\let\svgwidth\undefined%
  \global\let\svgscale\undefined%
  \makeatother%
  \begin{picture}(1,0.62036131)%
    \put(0,0){\includegraphics[width=\unitlength]{Cells.eps}}%
    \put(0.25334478,0.61467417){\color[rgb]{0,0,0}\makebox(0,0)[lb]{\smash{ }}}%
    \put(0.2534273,0.21142911){\color[rgb]{0,0,0}\makebox(0,0)[lb]{\smash{$k^\text{th}$ User}}}%
    \put(0.14027852,0.00100083){\color[rgb]{0,0,0}\makebox(0,0)[lb]{\smash{$l^\text{th}$ Cell}}}%
    \put(0.72355706,0.15831791){\color[rgb]{0,0,0}\makebox(0,0)[lb]{\smash{$j^\text{th}$ Cell}}}%
    \put(0.35135253,0.4771847){\color[rgb]{0,0,0}\makebox(0,0)[lb]{\smash{$\sqrt{\beta_{jkl}}h_{jkl}[m]$}}}%
    \put(0.79912333,0.49873393){\color[rgb]{0,0,0}\makebox(0,0)[lb]{\smash{$1$}}}%
    \put(0.79544055,0.39268692){\color[rgb]{0,0,0}\makebox(0,0)[lb]{\smash{$m$}}}%
    \put(0.79558778,0.28813793){\color[rgb]{0,0,0}\makebox(0,0)[lb]{\smash{$M$}}}%
  \end{picture}%
\endgroup%

%% file: region.eps_tex
\begingroup%
  \makeatletter%
  \providecommand\color[2][]{%
    \errmessage{(Inkscape) Color is used for the text in Inkscape, but the package 'color.sty' is not loaded}%
    \renewcommand\color[2][]{}%
  }%
  \providecommand\transparent[1]{%
    \errmessage{(Inkscape) Transparency is used (non-zero) for the text in Inkscape, but the package 'transparent.sty' is not loaded}%
    \renewcommand\transparent[1]{}%
  }%
  \providecommand\rotatebox[2]{#2}%
  \newcommand*\fsize{\dimexpr\f@size pt\relax}%
  \newcommand*\lineheight[1]{\fontsize{\fsize}{#1\fsize}\selectfont}%
  \ifx\svgwidth\undefined%
    \setlength{\unitlength}{1238.74015748bp}%
    \ifx\svgscale\undefined%
      \relax%
    \else%
      \setlength{\unitlength}{\unitlength * \real{\svgscale}}%
    \fi%
  \else%
    \setlength{\unitlength}{\svgwidth}%
  \fi%
  \global\let\svgwidth\undefined%
  \global\let\svgscale\undefined%
  \makeatother%
  \begin{picture}(1,0.87871854)%
    \lineheight{1}%
    \setlength\tabcolsep{0pt}%
    \put(0,0){\includegraphics[width=\unitlength]{region.eps}}%
    \put(0.04497657,0.53279938){\color[rgb]{0,0,0}\makebox(0,0)[lt]{\lineheight{1.25}\smash{\begin{tabular}[t]{l}A\end{tabular}}}}%
    \put(0.03978697,0.36154245){\color[rgb]{0,0,0}\makebox(0,0)[lt]{\lineheight{1.25}\smash{\begin{tabular}[t]{l}B\end{tabular}}}}%
    \put(0.44457627,0.01435772){\color[rgb]{0,0,0}\makebox(0,0)[lt]{\lineheight{1.25}\smash{\begin{tabular}[t]{l}C\end{tabular}}}}%
    \put(0.6694588,0.01608751){\color[rgb]{0,0,0}\makebox(0,0)[lt]{\lineheight{1.25}\smash{\begin{tabular}[t]{l}D\end{tabular}}}}%
    \put(0.01818937,0.82784271){\color[rgb]{0,0,0}\makebox(0,0)[lt]{\lineheight{0}\smash{\begin{tabular}[t]{l}$R_{i2}$\end{tabular}}}}%
    \put(0.89089509,0.04220962){\color[rgb]{0,0,0}\makebox(0,0)[lt]{\lineheight{0}\smash{\begin{tabular}[t]{l}$R_{i1}$\end{tabular}}}}%
    \put(0.20096943,0.58759024){\color[rgb]{0,0,0}\makebox(0,0)[lt]{\lineheight{0}\smash{\begin{tabular}[t]{l}$G_{1}$\end{tabular}}}}%
    \put(0.58963792,0.48079205){\color[rgb]{0,0,0}\makebox(0,0)[lt]{\lineheight{0}\smash{\begin{tabular}[t]{l}$G_{2}$\end{tabular}}}}%
    \put(0.71245859,0.18152486){\color[rgb]{0,0,0}\makebox(0,0)[lt]{\lineheight{0}\smash{\begin{tabular}[t]{l}$G_{3}$\end{tabular}}}}%
  \end{picture}%
\endgroup%

%% file: case_222.eps_tex
\begingroup%
  \makeatletter%
  \providecommand\color[2][]{%
    \errmessage{(Inkscape) Color is used for the text in Inkscape, but the package 'color.sty' is not loaded}%
    \renewcommand\color[2][]{}%
  }%
  \providecommand\transparent[1]{%
    \errmessage{(Inkscape) Transparency is used (non-zero) for the text in Inkscape, but the package 'transparent.sty' is not loaded}%
    \renewcommand\transparent[1]{}%
  }%
  \providecommand\rotatebox[2]{#2}%
  \ifx\svgwidth\undefined%
    \setlength{\unitlength}{785.19685039bp}%
    \ifx\svgscale\undefined%
      \relax%
    \else%
      \setlength{\unitlength}{\unitlength * \real{\svgscale}}%
    \fi%
  \else%
    \setlength{\unitlength}{\svgwidth}%
  \fi%
  \global\let\svgwidth\undefined%
  \global\let\svgscale\undefined%
  \makeatother%
  \begin{picture}(1,0.96028881)%
    \put(0,0){\includegraphics[width=\unitlength]{case_222.eps}}%
    \put(0.54351596,0.42031355){\color[rgb]{0,0,0}\makebox(0,0)[lb]{\smash{E}}}%
    \put(0.56645753,0.55326006){\color[rgb]{0,0,0}\makebox(0,0)[lb]{\smash{F}}}%
    \put(0.67255459,0.53666301){\color[rgb]{0,0,0}\makebox(0,0)[lb]{\smash{G}}}%
    \put(0.81537137,0.9215612){\color[rgb]{0,0,0}\makebox(0,0)[lb]{\tiny{\smash{SND}}}}%
    \put(0.81133548,0.87128332){\color[rgb]{0,0,0}\makebox(0,0)[lb]{\tiny{\smash{S-SND}}}}%
    \put(0.81662108,0.81757832){\color[rgb]{0,0,0}\makebox(0,0)[lb]{\tiny{\smash{SD}}}}%
    \put(0.81629971,0.76133613){\color[rgb]{0,0,0}\makebox(0,0)[lb]{\tiny{\smash{TIN}}}}%
    \put(-0.00006737,0.99939968){\color[rgb]{0,0,0}\makebox(0,0)[lb]{\smash{$R_{i2}$}}}%
    \put(0.99949591,0.01569227){\color[rgb]{0,0,0}\makebox(0,0)[lb]{\smash{$R_{i1}$}}}%
    \put(0.79945244,0.63380167){\color[rgb]{0,0,0}\makebox(0,0)[lb]{\smash{$R_{i2} = R_{i1}$}}}%
  \end{picture}%
\endgroup%

%% file: case_111.eps_tex
\begingroup%
  \makeatletter%
  \providecommand\color[2][]{%
    \errmessage{(Inkscape) Color is used for the text in Inkscape, but the package 'color.sty' is not loaded}%
    \renewcommand\color[2][]{}%
  }%
  \providecommand\transparent[1]{%
    \errmessage{(Inkscape) Transparency is used (non-zero) for the text in Inkscape, but the package 'transparent.sty' is not loaded}%
    \renewcommand\transparent[1]{}%
  }%
  \providecommand\rotatebox[2]{#2}%
  \ifx\svgwidth\undefined%
    \setlength{\unitlength}{785.19685039bp}%
    \ifx\svgscale\undefined%
      \relax%
    \else%
      \setlength{\unitlength}{\unitlength * \real{\svgscale}}%
    \fi%
  \else%
    \setlength{\unitlength}{\svgwidth}%
  \fi%
  \global\let\svgwidth\undefined%
  \global\let\svgscale\undefined%
  \makeatother%
  \begin{picture}(1,0.96028881)%
    \put(0,0){\includegraphics[width=\unitlength]{case_111.eps}}%
    \put(0.43746987,0.30156215){\color[rgb]{0,0,0}\makebox(0,0)[lb]{\smash{H}}}%
    \put(0.50524054,0.49337666){\color[rgb]{0,0,0}\makebox(0,0)[lb]{\smash{I}}}%
    \put(0.72537137,0.9115612){\color[rgb]{0,0,0}\makebox(0,0)[lb]{\tiny{\smash{SND}}}}%
    \put(0.72133548,0.86128332){\color[rgb]{0,0,0}\makebox(0,0)[lb]{\tiny{\smash{S-SND}}}}%
    \put(0.72662108,0.80357832){\color[rgb]{0,0,0}\makebox(0,0)[lb]{\tiny{\smash{SD}}}}%
    \put(0.72629971,0.74033613){\color[rgb]{0,0,0}\makebox(0,0)[lb]{\tiny{\smash{TIN}}}}%
    \put(0.89513496,0.64951869){\color[rgb]{0,0,0}\makebox(0,0)[lt]{\begin{minipage}{0.06003951\unitlength}\raggedright \end{minipage}}}%
    \put(-0.00006737,0.99939968){\color[rgb]{0,0,0}\makebox(0,0)[lb]{\smash{$R_{i2}$}}}%
    \put(0.99949591,0.01569227){\color[rgb]{0,0,0}\makebox(0,0)[lb]{\smash{$R_{i1}$}}}%
    \put(0.73945244,0.61180167){\color[rgb]{0,0,0}\makebox(0,0)[lb]{\smash{$R_{i2} = R_{i1}$}}}%
  \end{picture}%
\endgroup%

%% file: circles_2.eps_tex
\begingroup%
  \makeatletter%
  \providecommand\color[2][]{%
    \errmessage{(Inkscape) Color is used for the text in Inkscape, but the package 'color.sty' is not loaded}%
    \renewcommand\color[2][]{}%
  }%
  \providecommand\transparent[1]{%
    \errmessage{(Inkscape) Transparency is used (non-zero) for the text in Inkscape, but the package 'transparent.sty' is not loaded}%
    \renewcommand\transparent[1]{}%
  }%
  \providecommand\rotatebox[2]{#2}%
  \ifx\svgwidth\undefined%
    \setlength{\unitlength}{1183.62235482bp}%
    \ifx\svgscale\undefined%
      \relax%
    \else%
      \setlength{\unitlength}{\unitlength * \real{\svgscale}}%
    \fi%
  \else%
    \setlength{\unitlength}{\svgwidth}%
  \fi%
  \global\let\svgwidth\undefined%
  \global\let\svgscale\undefined%
  \makeatother%
  \begin{picture}(1,0.27018178)%
    \put(0,0){\includegraphics[width=\unitlength]{circles_2.eps}}%
    \put(0.32763191,0.08811876){\color[rgb]{0,0,0}\makebox(0,0)[lb]{\smash{}}}%
    \put(0.25580975,0.00181842){\color[rgb]{0,0,0}\makebox(0,0)[lb]{\smash{Cell $1$}}}%
    \put(0.50518897,0.00469897){\color[rgb]{0,0,0}\makebox(0,0)[lb]{\smash{Cell $2$}}}%
    \put(0.4095669,0.08932667){\color[rgb]{0,0,0}\makebox(0,0)[lb]{\smash{}}}%
    \put(0.41244697,0.21856166){\color[rgb]{0,0,0}\makebox(0,0)[lb]{\smash{$d$}}}%
    \put(0.25710481,0.22473073){\color[rgb]{0,0,0}\makebox(0,0)[lb]{\smash{$x$}}}%
  \end{picture}%
\endgroup%

%% file: three-cell-theta.eps_tex
\begingroup%
  \makeatletter%
  \providecommand\color[2][]{%
    \errmessage{(Inkscape) Color is used for the text in Inkscape, but the package 'color.sty' is not loaded}%
    \renewcommand\color[2][]{}%
  }%
  \providecommand\transparent[1]{%
    \errmessage{(Inkscape) Transparency is used (non-zero) for the text in Inkscape, but the package 'transparent.sty' is not loaded}%
    \renewcommand\transparent[1]{}%
  }%
  \providecommand\rotatebox[2]{#2}%
  \ifx\svgwidth\undefined%
    \setlength{\unitlength}{942.67738643bp}%
    \ifx\svgscale\undefined%
      \relax%
    \else%
      \setlength{\unitlength}{\unitlength * \real{\svgscale}}%
    \fi%
  \else%
    \setlength{\unitlength}{\svgwidth}%
  \fi%
  \global\let\svgwidth\undefined%
  \global\let\svgscale\undefined%
  \makeatother%
  \begin{picture}(1,0.46252683)%
    \put(0,0){\includegraphics[width=\unitlength]{three-cell-theta.eps}}%
    \put(0.09984595,0.00746454){\color[rgb]{0,0,0}\makebox(0,0)[lb]{\smash{Cell $1$}}}%
    \put(0.37386357,0.00824157){\color[rgb]{0,0,0}\makebox(0,0)[lb]{\smash{Cell $2$}}}%
    \put(0.63768689,0.01094775){\color[rgb]{0,0,0}\makebox(0,0)[lb]{\smash{Cell $3$}}}%
    \put(0.07962214,0.37393209){\color[rgb]{0,0,0}\makebox(0,0)[lb]{\smash{$x$}}}%
    \put(0.25917938,0.42021248){\color[rgb]{0,0,0}\makebox(0,0)[lb]{\smash{$2x$}}}%
    \put(0.91360588,0.25095587){\color[rgb]{0,0,0}\makebox(0,0)[lb]{\smash{$x$}}}%
    \put(0.35394067,0.21311318){\color[rgb]{0,0,0}\makebox(0,0)[lb]{\smash{$\theta$}}}%
  \end{picture}%
\endgroup%